\documentclass{article}
\usepackage{amssymb}
\usepackage{amsmath}
\usepackage{amsthm}
\usepackage{color}
\usepackage[colorlinks]{hyperref}
\usepackage[a4paper]{geometry}

\usepackage{amsfonts}
\usepackage{algorithm}
\usepackage{array}
\usepackage[caption=false,font=normalsize,labelfont=sf,textfont=sf]{subfig}
\usepackage{textcomp}
\usepackage{stfloats}
\usepackage{xurl}
\usepackage{verbatim}
\usepackage{graphicx}
\usepackage{cite}
\usepackage{algpseudocode}
\usepackage{placeins}
\usepackage{booktabs}
\usepackage{makecell}
\usepackage{longtable}

\newtheorem{theorem}{Theorem}[section]
\newtheorem{lemma}[theorem]{Lemma}

\theoremstyle{definition}

\theoremstyle{remark}
\newtheorem{remark}{Remark}

\newcommand{\wt}{\operatorname{wt}}

\newcommand{\Hull}{\operatorname{Hull}}

\title{Embedding Linear Codes into Optimal Self-Orthogonal Codes Based on Constraint-Preserving Genetic Algorithms}
\author{Haeun Lim\thanks{haeunlim@sogang.ac.kr, Department of Mathematics and Institute for Mathematical and Data Sciences, Sogang University, Seoul, Korea}, Junmin An\thanks{junmin0518@sogang.ac.kr, Department of Mathematics and Institute for Mathematical and Data Sciences, Sogang University, Seoul, Korea}, Jon-Lark Kim\thanks{jlkim@sogang.ac.kr, Department of Mathematics and Institute for Mathematical and Data Sciences, Sogang University, Seoul, Korea}}
\date{}

\begin{document}
\maketitle

\begin{abstract}
We study the problem of finding shortest self-orthogonal embeddings of a given code which have the largest possible minimum distance by employing a constraint-preserving genetic algorithm. We show that the orthogonal group is generated by transvections corresponding to weight-four vectors. Using this result together with the fact that all shortest self-orthogonal embeddings can be obtained through the action of the orthogonal group, we construct a genetic algorithm in which each chromosome is represented by a sequence of weight-four vectors. Applying our algorithm to codes with dimensions $9\le k\le 15$ and lengths $22\le n\le 256$, we obtain 280 inequivalent optimal self-orthogonal codes. In particular, we obtain 34 codes with new parameters and 100 codes which are inequivalent to previously known codes with the same parameters.\\

\end{abstract}

\section{Introduction}
One of the central topics in coding theory is to find optimal codes with the largest possible minimum distance for given parameters. Various bounds on the minimum distance for codes with have been established, and numerous methods have been developed to construct codes meeting these bounds~\cite{BLS-1974,C-2025,HXLZWT-2024,HT-1981,HV-2025,SS-1965}. Self-orthogonal codes are regarded as one of the most important classes of codes in coding theory, as they are closely related to various combinatorial structures, including graphs~\cite{FOS-2024}, $t$-designs~\cite{BG-2004,HKM-2004}, and lattices~\cite{BBH-2013,H-2009}, can be used to construct quantum codes via the CSS construction~\cite{CRSS-1998}, and include many codes with good parameters~\cite{ZKL-2025}. Therefore, in addition to finding optimal linear codes, finding optimal self-orthogonal codes with the largest possible minimum distance among self-orthogonal codes with given parameters is considered equally important. Accordingly, extensive research has been conducted to constructing optimal self-orthogonal codes~\cite{BBGO-2006,B-1999,KLL-2026,KW-2009,P-1972,P-1972-2,PS-1975,SLHK-2025}.

Recently, embedding methods have been actively studied as an approach to finding new distance-optimal codes. An embedding method constructs a new code satisfying a prescribed duality property by appending columns to a generator matrix of a given code. Two of the main topics in the study of embedding methods are determining the minimum number of columns that must be appended to achieve the desired duality property and developing methods for obtaining embeddings with the best possible minimum distance. The embedding method was first studied in~\cite{KKL-2021}, where an algorithm was proposed for constructing self-orthogonal codes from binary codes of dimension three or four by appending the minimum number of columns. Using this method, several new self-orthogonal codes were constructed. Subsequently, in~\cite{KC-2022}, this result was extended to binary codes of dimension up to six, and an upper bound was established on the minimum number of columns that must be appended to embed a binary code of arbitrary dimension into a self-orthogonal code. In~\cite{AKKLW-2025}, the problem of determining the minimum number of columns required for a self-orthogonal embedding of a binary code was completely solved for arbitrary dimension.

An et al.~\cite{AHKL-2026} extended the embedding method, which had previously been restricted to self-orthogonal codes, to LCD codes. They determined the minimum number of columns required to embed a given binary, ternary, or quaternary code into an LCD code and, using this approach, constructed several new LCD codes. Wang et al.~\cite{WL-2026} generalized this result by determining the minimum number of columns required to embed a given code over a finite field into a code with a prescribed hull dimension. Moreover, in~\cite{AKL-2026} and~\cite{AK-2026}, the embedding method was extended to codes over rings, where the minimum number of columns required for a shortest self-orthogonal embedding was investigated over $\mathbb{Z}_4$ and $\mathbb{F}_q+u\mathbb{F}_q$, respectively. However, despite these efforts to extend the embedding method to various hull dimensions and different rings, relatively little progress has been made in developing methods for finding embeddings with optimal minimum distance.

Therefore, we address this problem using a genetic algorithm. A genetic algorithm is a metaheuristic optimization method inspired by the processes of natural selection and genetic evolution, and was first introduced by Holland~\cite{H-1973}. Since then, genetic algorithms have been regarded as effective tools for solving complex optimization problems~\cite{BGH-1989,G-1986}. Korban et al.~\cite{KSU-2024} first introduced genetic algorithms into coding theory. Using this approach, they discovered previously unknown self-dual codes, thereby demonstrating the effectiveness of genetic algorithms for searching for new codes. More recently, Won et al.~\cite{WKL-2026} applied a modified genetic algorithm to the building-up construction introduced by Kim~\cite{K-2023} and obtained new LCD codes. These results motivate our application of a genetic algorithm to the embedding method.

In this paper, we develop a modified genetic algorithm that can be applied to the self-orthogonal embedding problem. First, we present an algorithm for constructing a shortest self-orthogonal embedding. We then use the fact that all other shortest self-orthogonal embeddings can be obtained by applying the action of the orthogonal group to the shortest self-orthogonal embedding constructed in this algorithm. We show that the orthogonal group of degree $\ge 5$ is generated by transvections corresponding to vectors of weight four. Using this, we define a chromosome as a sequence of weight-four vectors. We also define a fitness function based on the minimum distance of the code and the number of codewords attaining the minimum distance, and implement the genetic algorithm using two different types of crossover operations. When the number of columns appended for a shortest self-orthogonal embedding is less than five, we employ an exhaustive search algorithm. Otherwise, we use the genetic algorithm. Using this approach, we obtain 280 optimal self-orthogonal codes with parameters in the ranges $9\le k\le 15$ and $22\le n\le 256$, where $k$ and $n$ denote the dimension and length of the code, respectively. Among these, 34 codes have previously unknown parameters, while 100 codes are inequivalent to the previously known codes with the same parameters.

This paper is organized as follows. In Section 2, we briefly review basic concepts from coding theory and genetic algorithms. In Section 3, we introduce shortest self-orthogonal embeddings and develop a genetic algorithm for finding such embeddings. Section 4 presents the methodology, while Section 5 shows the results obtained by applying our genetic algorithm. We conclude our paper in Section 6.

\section{Preliminaries}
\subsection{Introduction to Coding Theory}
Let $\mathbb{F}_2$ be the finite field with two elements. A \textit{code} $\mathcal{C}$ of length $n$ over $\mathbb{F}_2$ is a subset of $\mathbb{F}_2^n$ for $n \geq 1$. Elements of $\mathcal{C}$ are called \textit{codewords}. A $k$-dimensional subspace of $\mathbb{F}_2^n$ is called a \textit{linear code}, and is denoted by an $[n,k]$ code. For an $[n,k]$ code  $\mathcal{C}$, a $k \times n$ matrix $G$ over $\mathbb{F}_2$ whose rows form a basis for $\mathcal{C}$ is called a \textit{generator matrix} of $\mathcal{C}$.

Two codes $\mathcal{C}_1$ and $\mathcal{C}_2$ over $\mathbb{F}_2$ are called \textit{equivalent} if there exists a permutation of columns $\sigma$ such that $\sigma \mathcal{C}_1 = \mathcal{C}_2$. The \textit{automorphism group} of a code $\mathcal{C}$, denoted by $\mathrm{Aut}(\mathcal{C})$, is the group of all coordinate permutations that fix $\mathcal{C}$, that is, $\mathrm{Aut}(\mathcal{C}) = \{\sigma \in S_n ~|~ \sigma\mathcal{C} = \mathcal{C}\}$, where $S_n$ is the symmetric group on the $n$ coordinate positions.

The \textit{dual} $\mathcal{C}^{\perp}$ of a code $\mathcal{C}$ is defined as the set
\[
\mathcal{C}^\perp=\{\mathbf{x} \in \mathbb{F}_2^n~|~\mathbf{x} \cdot \mathbf{c} = 0  \textrm{ for all } \mathbf{c} \in \mathcal{C} \}
\]
where $\cdot$ is the standard dot product. If $\mathcal{C} \subseteq \mathcal{C}^{\perp}$, then $\mathcal{C}$ is called \textit{self-orthogonal}. Moreover, if $\mathcal{C} = \mathcal{C}^\perp$, then it is called \textit{self-dual}. The \textit{hull} of $\mathcal{C}$, denoted by $\mathrm{Hull}(\mathcal{C})$, is defined as \[
\mathrm{Hull}(\mathcal{C}) = \mathcal{C} \cap \mathcal{C}^\perp.
\]
It is straightforward that $\mathcal{C}$ is self-orthogonal if and only if $\mathcal{C}=\mathrm{Hull}(\mathcal{C})$. For an $[n, k]$ code $\mathcal{C}$ over $\mathbb{F}_2$, it is known from~\cite{GJG-2018} that the dimension of $\mathrm{Hull}(\mathcal{C})$ is given as
\[
\dim\mathrm{Hull}(\mathcal{C})=k-\operatorname{rank}(GG^T),
\]
where $G$ is a generator matrix of $\mathcal{C}$. It follows that $\mathcal{C}$ is self-orthogonal if and only if $GG^{T}=\mathcal{O}$.

The \textit{Hamming weight} of a vector $\mathbf{x} \in \mathbb{F}_2^n$, denoted by $\wt(\mathbf{x})$, is the number of nonzero coordinates in $\mathbf{x}$. For two vectors $\mathbf{x}, \mathbf{y}$ in $\mathbb{F}_2^n$, the \textit{Hamming distance} between $\mathbf{x}$ and $\mathbf{y}$ is defined as
\[
d(\mathbf{x}, \mathbf{y})=\wt(\mathbf{x} - \mathbf{y}).
\]
The \textit{minimum (Hamming) distance} of a code $\mathcal{C}$, denoted by $d(\mathcal{C})$, is the minimum of the distances between any two distinct codewords in $\mathcal{C}$. An $[n,k]$ code $\mathcal{C}$ with minimum distance $d$ is denoted by an $[n,k,d]$ code. For a code $\mathcal{C}$ over $\mathbb{F}_2$, if every codeword of $\mathcal{C}$ has even weight, then $\mathcal{C}$ is called an \textit{even code}. Otherwise, it is called an \textit{odd-like code}.
Given $n$ and $k$, we denote by $d(n, k)$ the largest possible minimum distance over all linear $[n, k]$ codes. A linear $[n,k]$ code $\mathcal{C}$ with minimum distance $d(\mathcal{C})$ is called \textit{distance-optimal} if
\[
d(\mathcal{C})=d(n,k).
\]
Similarly, $d_{SO}(n, k)$ is the largest possible minimum distance over all self-orthogonal $[n, k]$ codes. A self-orthogonal $[n,k]$ code $\mathcal{C}$ is called an \textit{distance-optimal self-orthogonal code} if
\[
d(\mathcal{C}) = d_{SO}(n, k).
\]

For an $[n,k]$ code $\mathcal{C}$, let $S$ be a subset of the set of coordinate positions of $\mathcal{C}$. The code obtained by removing the columns indexed by $S$ from a generator matrix $G$ of $\mathcal{C}$ is called a \textit{punctured code} of $\mathcal{C}$. For a codeword $\mathbf{c}$ of a code $\mathcal{C}$ over $\mathbb{F}_2$, the residual code $\mathrm{Res}(\mathcal{C}, \mathbf{c})$ is defined as the code obtained by puncturing all coordinate positions where $\mathbf{c}$ has ones.

\subsection{Genetic Algorithm}
A genetic algorithm is a population-based heuristic search algorithm inspired by the mechanisms of natural evolution and adaptation. It is used to find good approximate solutions to optimization problems, especially when the search space is too large for exhaustive search. In a genetic algorithm, candidate solutions are represented as artificial \textit{chromosomes}, and each chromosome is evaluated by a \textit{fitness function}, which measures how promising the chromosome is as a solution. For a detailed introduction to genetic algorithms, we refer the reader to~\cite{G-1989}.
The basic idea of a genetic algorithm is that high-quality ``parent'' candidate solutions from different regions in the space may produce high-quality ``offspring'' candidate solutions. Based on this idea, chromosomes with higher fitness values are more likely to be selected as parents. The selected parents then generate new chromosomes through \textit{crossover}, which combines parts of two parent chromosomes. In addition, \textit{mutation} is applied to introduce small random changes to the offspring chromosomes and to maintain diversity in the population.
We now describe the main steps of a genetic algorithm in more detail.

\textbf{Chromosomes and genes.}
A chromosome represents a candidate solution to the given problem. It is often encoded as a bit string, although the encoding may vary depending on the problem. The basic components of a chromosome are called \textit{genes}. In a bit-string representation, genes may be single bits or short blocks of adjacent bits encoding particular elements of the candidate solution.

\textbf{Selection.}
Selection is the process of choosing chromosomes from the current population to serve as parents for the next generation. We also use \textit{elitism}, which preserves the best chromosomes from the previous generation. After preserving the elite chromosomes, the remaining positions in the next generation are filled through selection. For the parent selection mechanism, we employ \textit{tournament selection}. In each selection event, $t$ candidate chromosomes are randomly drawn from the population, and the individuals with high fitness function values among them are selected. This procedure is iteratively executed until the required number of parents is collected.

\textbf{Crossover.}
Crossover combines two parent chromosomes to produce offspring. In one-point crossover, a crossover point is chosen, and the subsequences before and after this point are exchanged between the two parents. For example, if the two parent strings are
\[
10001000
\qquad\text{and}\qquad
11111111,
\]
then crossing over after the third position produces
\[
10011111
\qquad\text{and}\qquad
11101000.
\]
This operation roughly mimics biological recombination. If each pair of parents produces two offspring and the population is replaced by the same number of chromosomes, then the population size remains unchanged.

\textbf{Mutation.}
Mutation introduces random changes into a chromosome. In the bit-string representation, mutation is often implemented by flipping randomly selected bits. More generally, mutation operations can be designed according to the structure of the problem.

The sequence of selection, crossover, and mutation is iterated for a predetermined number of \textit{generations}, or until a prescribed stopping criterion is satisfied.
In many applications of genetic algorithms, chromosomes can be generated and modified rather freely. However, in some problems, not every formal chromosome represents a feasible solution. In such cases, the chromosome must satisfy certain structural constraints, and the genetic operations must be designed so that these constraints are preserved after crossover and mutation.

\section{Shortest self-orthogonal embeddings and a Constraint-Preserving Genetic Algorithm}
In this section, we introduce the shortest self-orthogonal embedding of binary linear codes, and present a modified genetic algorithm suitable for our problem.

For an $[n, k]$ code $\mathcal{C}$ over $\mathbb{F}_2$, an $[n', k]$ self-orthogonal code $\mathcal{C}'$ is called a \textit{self-orthogonal embedding} of $\mathcal{C}$ over $\mathbb{F}_2$ if $\mathcal{C}$ is a punctured code of $\mathcal{C}'$. Let $m = n'-n$. Up to equivalence, we may assume that puncturing $\mathcal{C}'$ at the last $m$ coordinates yields $\mathcal{C}$. Thus, we consider a generator matrix of $\mathcal{C}'$ of the form \[
G' = [G~|~A].
\]
where $G$ is a generator matrix of $\mathcal{C}$. A \textit{shortest self-orthogonal embedding} of $\mathcal{C}$ is a self-orthogonal embedding of $\mathcal{C}$ whose length is the smallest among all self-orthogonal embeddings of $\mathcal{C}$. Using the notion of the hull of a code $\mathcal{C}$, An et al. determined the length of a shortest self-orthogonal embedding of $\mathcal{C}$.
\begin{theorem}[\cite{AKKLW-2025}]\label{AKKLW-2025}
    Let $\mathcal{C}$ be a binary $[n,k]$ code and let $\ell=\dim \Hull(\mathcal{C})$. Then the length $n'$ of a shortest self-orthogonal embedding of $\mathcal{C}$ is
    \[
    n' = \begin{cases}
        n+k-\ell+1,   &\mathrm{if~}\mathcal{C}\textup{~is~even,}\\
        n+k-\ell,     &\mathrm{if~}\mathcal{C}\textup{~is~odd-like.}
    \end{cases}
    \]
\end{theorem}
Let $\mathcal{C}$ be a code over $\mathbb{F}_2$. Let $H$ be a generator matrix of $\Hull(\mathcal{C})$. Since $\Hull(\mathcal{C})$ is a subcode of $\mathcal{C}$, we can extend $H$ to a generator matrix $G$ of $\mathcal{C}$ as follows:
\[
G(\mathcal{C}) = \begin{bmatrix}
    H \\ A
\end{bmatrix}
\]
for some matrix $A$. The following theorem provides a characterization of all shortest self-orthogonal embeddings of $\mathcal{C}$.

\begin{theorem}[\cite{AKKLW-2025}]
Let $\mathcal{C}$ be a code over $\mathbb{F}_2$ with a generator matrix $G$ of the form
\[
G = \begin{bmatrix}
    H \\ A
\end{bmatrix}
\] where $H$ is a generator matrix of $\Hull(\mathcal{C})$. Then, a shortest self-orthogonal embedding $\mathcal{C}'$ of $\mathcal{C}$ has a generator matrix
\[
G'=\begin{bmatrix}
    H & O\\
    A & B
    \end{bmatrix}
\]
for some matrix $B$.
\end{theorem}

\begin{remark}\label{base_SO_emb_alg}
Here, we provide an algorithm to construct a shortest self-orthogonal embedding of an $[n, k]$ code $\mathcal{C}$ with a generator matrix of the form
\[
G = \begin{bmatrix}
    H \\ A
\end{bmatrix}
\] where $H$ is a generator matrix of $\Hull(\mathcal{C})$. Let $r=k-\dim \Hull(\mathcal{C})$.

\medskip

First, we assume that $\mathcal{C}$ is an odd-like code. Since the rows of $A$ generate an LCD code, $AA^T$ is a non-singular symmetric matrix. Moreover, since $\mathcal{C}$ is odd-like, $AA^T$ is non-alternating, that is, it has at least one nonzero diagonal entry. We want to construct an $r\times r$ matrix $B$ such that $AA^T=BB^T$ as follows. Let $P=I_r$.
\begin{enumerate}
    \item \textbf{Step 1.} Since $AA^T$ is non-alternating, there exists $p\in\{1, \ldots, r\}$ such that $(AA^T)_{pp}=1$. Interchange the first and $p$-th rows and the first and $p$-th columns of $AA^T$, respectively, and interchange the first and $p$-th rows of $P$.
    \item \textbf{Step 2.} For each $j>1$ such that $(AA^T)_{1j}=1$, add the first row of $AA^T$ to its $j$-th row and the first column of $AA^T$ to its $j$-th column. At the same time, add the first row of $P$ to its $j$-th row.
    \item \textbf{Step 3.} Repeat Steps 1 and 2 successively, starting from the second diagonal entry and proceeding through the remaining diagonal entries, while applying the corresponding row operations to $P$ at each step. Continue this process until all nonzero diagonal entries precede the zero diagonal entries. Consequently, $AA^T$ is transformed into a matrix of the form
    \[
    \begin{bmatrix}
        I_s & O\\O & A'
    \end{bmatrix}
    \]
    for some $s$ and a $(r-s)\times (r-s)$ non-singular alternating matrix $A'$. If $s=r$, skip the remaining steps.
    \item \textbf{Step 4.} Since $AA^T$ is non-singular, there exists $p,\,q\in \{s+1, \ldots, r\}$ such that $(AA^T)_{pq}=1$. Interchange the $p$-th and $q$-th rows of $AA^T$ with its $(s+1)$-st and $(s+2)$-nd rows, respectively, and interchange the corresponding columns in the same way. At the same time, interchange the $p$-th and $q$-th rows of $P$ with its $(s+1)$-st and $(s+2)$-nd rows, respectively.
    \item \textbf{Step 5.} For each $j>s+2$ such that $(AA^T)_{(s+1)j}=1$, add the $(s+2)$-nd row of $AA^T$ to its $j$-th row and the $(s+2)$-nd column of $AA^T$ to its $j$-th column. At the same time, add the $(s+2)$-nd row of $P$ to its $j$-th row. Similarly, for each $j>s+2$ such that $(AA^T)_{(s+2)j}=1$, add the $(s+1)$-st row of $AA^T$ to its $j$-th row and the $(s+1)$-st column of $AA^T$ to its $j$-th column. At the same time, add the $(s+1)$-st row of $P$ to its $j$-th row.
    \item \textbf{Step 6.} Repeat Steps 4 and 5 successively on the remaining submatrix, while applying the corresponding row operations to $P$ at each step. Continue this process until no entries remain. Consequently, $AA^T$ is transformed into a block diagonal matrix
    \[
    I_s\oplus \underbrace{H_2\oplus\ldots\oplus H_2}_{(r-s)/2}
    \]
    where
    \[
    H_2=\begin{bmatrix}
        0 & 1\\1 & 0
    \end{bmatrix}.
    \]
    \item \textbf{Step 7.} Let
    \[
    Q=
    \begin{bmatrix}
    1&1&1\\
    1&0&1\\
    1&1&0
    \end{bmatrix}.
    \]
    Consider the submatrix of $AA^T$ corresponding to the $s$-th, $(s+1)$-st, and $(s+2)$-nd rows and columns, which is
    \[
    M=
    \begin{bmatrix}
    1&0&0\\
    0&0&1\\
    0&1&0
    \end{bmatrix}.
    \]
    Then we have $QMQ^T=I_3$. Let $Q'$ be the $r\times r$ matrix whose submatrix corresponding to the $s$-th, $(s+1)$-st, and $(s+2)$-nd rows and columns is $Q$, whose remaining diagonal entries are all $1$, and whose remaining entries are all $0$. Replace $AA^T$ by $Q'AA^T(Q')^T$ and $P$ by $Q'P$.
    \item \textbf{Step 8.} Repeat Step 7 successively for the remaining $H$-blocks, while applying the corresponding transformations to $P$. After repeating this step, $AA^T$ is transformed into $I_r$. Denote the resulting matrix $P$ by $E$.
\end{enumerate}

Consequently, for the original matrix $AA^T$ before the transformations, we have $EAA^TE^T=I_r$, that is $AA^T=(E^{-1})(E^{-1})^T$. Define $B=E^{-1}$. Then
\[
\begin{bmatrix}
H & O\\
A & B
\end{bmatrix}
\]
generates a shortest self-orthogonal embedding of $\mathcal{C}$.

\medskip

Second, we assume that $\mathcal{C}$ is an even code. Then $AA^T$ is a non-singular alternating matrix. We want to construct an $r\times (r+1)$ matrix $B$ such that $AA^T=BB^T$ as follows.
\begin{enumerate}
    \item \textbf{Step 1.} Let $P=I_r$. Repeat Steps 4 and 5 from the odd-like case, while applying the same row operations to $P$. After these repetitions, $AA^T$ is transformed into
    \[
    \underbrace{H_2\oplus\ldots\oplus H_2}_{r/2},
    \]
    where
    \[
    H_2=\begin{bmatrix}0&1\\1&0\end{bmatrix}.
    \]
    We denote the resulting matrix $P$ by $E_1$. For the original matrix $AA^T$ before the transformations, we have
    \[
    E_1AA^TE_1^T=\underbrace{H_2\oplus\ldots\oplus H_2}_{r/2}.
    \]
    \item \textbf{Step 2.} Define $K=I_r+J_r$ where $J_r$ is the $r\times r$ all-one matrix. Since $K$ is also a non-singular alternating matrix, applying the same procedure as in Step 1 yields an invertible matrix $E_2\in GL_r(\mathbb{F}_2)$ such that
    \[
    E_2KE_2^T=\underbrace{H_2\oplus\ldots\oplus H_2}_{r/2}.
    \]
\end{enumerate}

Note that $K=[I_r~|~\mathbf{1}][I_r~|~\mathbf{1}]^T$ where $\mathbf{1}$ is the all-one vector of length $r$. Since
\[
E_1AA^TE_1^T=\underbrace{H_2\oplus\ldots\oplus H_2}_{r/2}=E_2KE_2^T,
\]

we have
\[
AA^T=(E_1^{-1}E_2)K(E_1^{-1}E_2)^T.
\]

Define $B=E_1^{-1}E_2[I_r~|~\mathbf{1}]$. Then
\[
\begin{bmatrix}
H & O\\
A & B
\end{bmatrix}
\]
generates a shortest self-orthogonal embedding of $\mathcal{C}$.
\end{remark}

\begin{theorem}[\cite{AKKLW-2025}]\label{zero-embedding}
Let $\mathcal{C}$ be a code over $\mathbb{F}_2$ with a generator matrix $G$ of the form
\[
G = \begin{bmatrix}
    H \\ A
\end{bmatrix}
\] where $H$ is a generator matrix of $\Hull(\mathcal{C})$, and let $\mathcal{C}'$ be a shortest self-orthogonal embedding of $\mathcal{C}$ with a generator matrix
\[
G'=\begin{bmatrix}
    H & O\\
    A & B
    \end{bmatrix}
\]
for some matrix $B$. Then, a code generated by
\[
G''=\begin{bmatrix}
    H & O\\
    A & BD
    \end{bmatrix},
\]
 is a shortest self-orthogonal embedding of $\mathcal{C}$, where $D$ is an orthogonal matrix. Moreover, every shortest self-orthogonal embedding of $\mathcal{C}$ has a generator matrix of this form for some orthogonal matrix $D$.
\end{theorem}
Therefore, by the algorithm given in Remark~\ref{base_SO_emb_alg}, we can construct a shortest self-orthogonal embedding of a code $\mathcal{C}$, and Theorem~\ref{zero-embedding} allows us to construct all other shortest self-orthogonal embeddings of $\mathcal{C}$ from it.

For $m \ge 1$, let $\mathbf{u}$ be a vector in $\mathbb{F}_2^m$. We define the \textit{transvection} $T_\mathbf{u}:\mathbb{F}_2^m \to \mathbb{F}_2^m$ as \[
T_\mathbf{u}(\mathbf{x})=\mathbf{x}+(\mathbf{x} \cdot \mathbf{u})\mathbf{u}=\mathbf{x}(I_m+\mathbf{u}^T\mathbf{u})
\] for every $\mathbf{x} \in \mathbb{F}_2^m$. From now on, we identify $T_\mathbf{u}$ with $I_m+\mathbf{u}^T\mathbf{u}$.
\begin{remark}\label{orthogonal-group}
    It is known from~\cite{J-2007} that the orthogonal group $\mathcal{O}_m$ of degree $m$ over $\mathbb{F}_2$ is given as follows:
    \begin{enumerate}
    \item if $1 \leq m \leq 3$, then $\mathcal{O}_m = \mathcal{P}_m$ where $\mathcal{P}_m$ is the group of all $m\times m$ permutation matrices.
    \item if $m \geq 4$, then $\mathcal{O}_m=\langle\mathcal{P}_m, T_{\mathbf{u}}\rangle$
    where $\mathbf{u}$ is a vector of Hamming weight 4 over $\mathbb{F}_2$.
    \end{enumerate}
\end{remark}
\begin{theorem}\label{orthogonal-group-generator}
    For $m \geq 5$, $\mathcal{O}_m$ is generated by \[
    \mathcal{T} = \{T_\mathbf{u}~|~\wt(\mathbf{u})=4~for~\mathbf{u}\in\mathbb{F}_2^m\}.
    \]
\end{theorem}
\begin{proof}
    For a weight four vector $\mathbf{u}_0 \in \mathbb{F}_2^m$, by Remark~\ref{orthogonal-group}, $\mathcal{O}_m=\langle\mathcal{P}_m, T_{\mathbf{u}_0}\rangle$. Thus, it is straightforward that $T_{\mathbf{u}_0} \in \mathcal{O}_m$.
    Next, we show that $\mathcal{T}$ generates $\mathcal{O}_m$, i.e., $\mathcal{P}_m \subseteq \langle\mathcal{T}\rangle$. Since $\mathcal{P}_m$ is generated by transpositions, it is sufficient to prove that any transposition $(ab)$ is generated by transvections in $\mathcal{T}$. We will show that any transposition can be written as a product of three transvections. Choose three distinct indices $c, d, e \notin \{a, b\}$, and set\[
    p=\mathbf{e}_a + \mathbf{e}_c + \mathbf{e}_d + \mathbf{e}_e,\quad q=\mathbf{e}_b + \mathbf{e}_c + \mathbf{e}_d + \mathbf{e}_e.
    \] where $\mathbf{e}_i$ denotes the $i$-th standard basis vector of $\mathbb{F}_2^m$, that is, the vector whose $i$-th coordinate is $1$ and all other coordinates are $0$.
    Then $\wt(p)=\wt(q)=4$ and $p \cdot q = 1$. For $\mathbf{x} \in \mathbb{F}_2^m$, we have
\begin{align*}
    T_p T_q T_p(\mathbf{x}) &= T_p T_q(\mathbf{x} + (\mathbf{x} \cdot p)p)\\
    &= T_p (\mathbf{x} + (\mathbf{x} \cdot p)p + (\mathbf{x} \cdot q + \mathbf{x} \cdot p)q)\\
    &= \mathbf{x} + (\mathbf{x} \cdot p)p + (\mathbf{x} \cdot q + \mathbf{x} \cdot p)q + (\mathbf{x} \cdot q)p\\
    &= \mathbf{x} + (\mathbf{x} \cdot (p+q))(p+q)\\
    &= T_{p+q}(\mathbf{x})\\
    &= T_{\mathbf{e}_a+\mathbf{e}_b}(\mathbf{x}).
\end{align*}
For $\mathbf{x} \in \mathbb{F}_2^m$,
\begin{align*}
    T_{\mathbf{e}_a+\mathbf{e}_b}(\mathbf{x})&=\mathbf{x} + (\mathbf{x} \cdot (\mathbf{e}_a+\mathbf{e}_b))(\mathbf{e}_a+\mathbf{e}_b)\\
    &=\mathbf{x} + (x_a+x_b)(\mathbf{e}_a+\mathbf{e}_b)\\
    &=\mathbf{x} + (x_a+x_b)\mathbf{e}_a + (x_a+x_b)\mathbf{e}_b\\
\end{align*}
Hence, we conclude that $T_{\mathbf{e}_a+\mathbf{e}_b}$ interchanges the $a$-th and $b$-th coordinates of $\mathbf{x}$ and fixes all other coordinates.
Thus, $T_p T_q T_p = T_{\mathbf{e}_a+\mathbf{e}_b} = (ab)$.
\end{proof}

It is known from~\cite{M-1969} that the size of the orthogonal group of degree $m$ is given as
\[
|\mathcal{O}_m|=
\begin{cases}
\displaystyle 2^{t}\prod_{i=1}^{t-1}\bigl(2^{2t}-2^{2i}\bigr), & \text{if } m=2t,\\[3mm]
\displaystyle 2^{t}\prod_{i=0}^{t-1}\bigl(2^{2t}-2^{2i}\bigr), & \text{if } m=2t+1,
\end{cases}
\]
where $t=\lfloor m/2\rfloor$. Thus, the size of the orthogonal group grows rapidly with its degree. By Theorem~\ref{zero-embedding}, the computational complexity of exhaustively examining all shortest self-orthogonal embeddings of a code $\mathcal{C}$ is on the size of the orthogonal group of degree $\dim(\mathcal{C})-\dim\Hull(\mathcal{C})$ or $\dim(\mathcal{C})-\dim\Hull(\mathcal{C})+1$. Consequently, such an exhaustive search becomes computationally infeasible when the dimension of $\mathcal{C}$ is large while the dimension of its hull is relatively small.

To address this problem, we apply a genetic algorithm to search for self-orthogonal embeddings, based on Theorem~\ref{orthogonal-group-generator}. Assume that the number of appended columns is $m\geq 5$. Let
\[
U=\{\mathbf{u}_1,\mathbf{u}_2,\ldots,\mathbf{u}_r\}
\]
be the set of all weight-$4$ vectors in $\mathbb{F}_2^m$.
We define a chromosome to be a finite sequence of vectors from $U$. More precisely, for a positive integer $s$, a chromosome is a sequence
\[
\chi=(\mathbf{u}_{i_1},\mathbf{u}_{i_2},\ldots,\mathbf{u}_{i_s}),
\]
where $\mathbf{u}_{i_j}\in U$ for each $j=1,\ldots,s$. We define the length of $\chi$ to be the number of weight-4 vectors represented by the sequence, namely $|\chi|=s$. For each chromosome $\chi$, we associate the matrix
\[
D_{\chi}
=
T_{\mathbf{u}_{i_1}}
T_{\mathbf{u}_{i_2}}
\cdots
T_{\mathbf{u}_{i_s}}.
\]
By Theorem~\ref{orthogonal-group-generator}, $D_{\chi}\in \mathcal{O}_m$.

Let $\mathcal{C}$ be a shortest self-orthogonal embedding of a code $\mathcal{C}_0$ over $\mathbb{F}_2$ with generator matrix
\[
G=\begin{bmatrix}
    H & O\\A & B
\end{bmatrix}
\]
where $H$ generates $\operatorname{Hull}(\mathcal{C}_0)$. For any chromosome $\chi$, we denote by $\mathcal{C}_{\chi}$ the code generated by
\[
G_{\chi}=
\begin{bmatrix}
H & O\\
A & BD_{\chi}
\end{bmatrix}.
\]
It is straightforward that $\mathcal{C}_\chi$ is a shortest self-orthogonal embedding of $\mathcal{C}_0$.

Before describing the genetic operations, we define the fitness function used to evaluate a chromosome. The primary objective is to maximize the minimum distance of $\mathcal{C}_{\chi}$. However, many chromosomes may yield codes with the same minimum distance. To distinguish such chromosomes, we use a \textit{lexicographic fitness function}. Let
\[
d_{\chi}=d(\mathcal{C}_{\chi})
\]
be the minimum distance of $\mathcal{C}_{\chi}$, and let
\[
A_i(\mathcal{C}_{\chi})
=
\left|\left\{
\mathbf{c}\in \mathcal{C}_{\chi}:\operatorname{wt}(\mathbf{c})=i\right\}\right|
\]
be the number of codewords of weight $i$ in $\mathcal{C}_{\chi}$. We define the \textit{fitness} of $\chi$ by
\[
F(\chi)
=
\left(
d_{\chi},
-A_{d_{\chi}}(\mathcal{C}_{\chi}),
-|\chi|
\right),
\]
where $|\chi|$ denotes the length of the chromosome. Fitness values are compared lexicographically. Thus, a chromosome that yields a code with a larger minimum distance is preferred. If two chromosomes yield codes with the same minimum distance, then the one with fewer minimum-weight codewords is preferred. The last component is considered only when the first two components are equal, in order to discourage unnecessarily long representations and improve computational efficiency.

We next define two types of crossover operators used in the genetic algorithm, referred to as the random crossover and the guided crossover. Let
\[
\chi_1=(\mathbf{u}_{i_1}, \mathbf{u}_{i_2}, \ldots, \mathbf{u}_{i_s}) \qquad \text{and} \qquad \chi_2=(\mathbf{u}_{j_1}, \mathbf{u}_{j_2}, \ldots, \mathbf{u}_{j_{s'}})
\]
be two chromosomes. First, we introduce the random crossover operation. Crossover points $a$ and $b$ are selected randomly within $1 \leq a \leq s$ and $1 \leq b \leq s'$. Then the output chromosome is \[
\chi_3 = (\mathbf{u}_{i_1}, \ldots, \mathbf{u}_{i_a},\mathbf{u}_{j_{b+1}}, \ldots, \mathbf{u}_{j_{s'}}),
\]
and
\[
\chi_4 = (\mathbf{u}_{j_1}, \ldots, \mathbf{u}_{j_b},\mathbf{u}_{i_{a+1}}, \ldots, \mathbf{u}_{i_{s}}).
\]
Next, we introduce the guided crossover operation. We regard $\chi_1$ as the primary chromosome and $\chi_2$ as the secondary chromosome. For each $1\leq r\leq s$, define
\[
{\chi_1}_r=(\mathbf{u}_{i_1},\ldots,\mathbf{u}_{i_r}).
\]
For a positive integer $p\leq s$, we select the $p$ \textit{prefix chromosomes} among $\chi_{1_1},\chi_{1_2},\ldots,\chi_{1_s}$ with the largest fitness values. Similarly, for each $1\leq r'\leq s'$, define
\[
{\chi_2}^{r'}=(\mathbf{u}_{j_{r'}},\ldots,\mathbf{u}_{j_{s'}}).
\]
For a positive integer $q\leq s'$, we select the $q$ \textit{suffix chromosomes} among ${\chi_2}^1,{\chi_2}^2,\ldots,{\chi_2}^{s'}$ with the largest fitness values. We then concatenate each selected prefix chromosome from the primary chromosome with each selected suffix chromosome from the secondary chromosome. In this way, we obtain $pq$ candidate offspring chromosomes. Among these candidates, we choose two chromosomes $\chi_3$ and $\chi_4$ with the largest fitness values and use them as offspring.

To prevent excessive chromosome length, we constrain the maximum length of the resulting chromosomes for both crossover operators. The maximum allowable length of a chromosome is set as a hyperparameter, denoted by $L_{\max}$. If $|\chi_3|>L_{\max}$ or $|\chi_4|>L_{\max}$, then the crossover points are chosen again.

We define several types of mutation operations. Let
\[
\chi=(\mathbf{u}_{i_1},\mathbf{u}_{i_2},\ldots,\mathbf{u}_{i_s})
\]
be a chromosome.
\begin{enumerate}
    \item \textbf{Replace.} Replace a randomly selected entry $\mathbf{u}_{i_\ell}$ of $\chi$ with a new vector $\mathbf{u}_{i_t}\in U$.

    \item \textbf{Insert.} Insert a new vector $\mathbf{u}_{i_t}\in U$ at a randomly selected position in $\chi$.

    \item \textbf{Delete.} Delete a randomly selected entry $\mathbf{u}_{i_\ell}$ from $\chi$.

    \item \textbf{Swap.} Interchange two randomly selected entries $\mathbf{u}_{i_{\ell_1}}$ and $\mathbf{u}_{i_{\ell_2}}$ of $\chi$. Since the orthogonal group $\mathcal{O}_m$ is generally non-abelian, changing the order of transvections may produce a different corresponding orthogonal matrix.
\end{enumerate}
We require every chromosome to be nonempty. Accordingly, the chromosome length is constrained by $1\leq |\chi|\leq L_{\max}$. To preserve the nonemptiness and maximum-length constraints during mutation, the set of admissible mutation operators depends on the current chromosome length $|\chi|$. If $|\chi|=1$, only replace and insert are permitted. If $2\leq |\chi|<L_{\max}$, all four mutation operators are permitted. If $|\chi|=L_{\max}$, insert is excluded. Moreover, since two consecutive identical transvections compose to the identity, both transvections are deleted whenever such a pair occurs.

\section{Methodology}

We first construct bounds on the minimum distances of self-orthogonal codes over $\mathbb{F}_2$ with dimension $9 \leq k \leq 15$. Grassl's table~\cite{G-online} provides bounds for $d(n,k)$. Clearly, this gives an upper bound for $d_{SO}(n,k)$. Shi et al. gave a bound for self-orthogonal codes using residual codes.
\begin{lemma}[\cite{SLHK-2025}]\label{shi-res}
    Let $\mathcal{C}$ be a self-orthogonal $[n,k,d]$ code over $\mathbb{F}_2$ and let $\mathbf{c}$ be a codeword of weight $w < 2d$, where $k \geq 2$. Then the residual code $\mathrm{Res}(\mathcal{C},\mathbf{c})$ of $\mathcal{C}$ is an even code with parameters
    \[
    \left[n-w, k-1, d' \geq 2 \left\lceil\frac{d}{2}-\frac{w}{4}\right\rceil\right].
    \] In particular, if $\mathbf{c}$ has weight $d$, where $d \equiv 2 \pmod4$, then $\mathrm{Res}(\mathcal{C},\mathbf{c})$ has parameters \[
    \left[ n-d, k-1, d' \geq \frac{d}{2}+1 \right].
    \]
\end{lemma}
Given $n$ and $k \geq 2$, let $d=d(n,k)$. Assume that there is a self-orthogonal $[n,k,d]$ code over $\mathbb{F}_2$. For a codeword $\mathbf{c} \in \mathcal{C}$ of weight $d$, we consider the residual code $\mathrm{Res}(\mathcal{C},\mathbf{c})$ with parameters \[
\left[n-d, k-1, d' \geq 2 \left\lceil\frac{d}{4}\right\rceil\right].
\]
If no even code with these parameters exists, then $d_{SO}(n,k)$ cannot be equal to $d(n,k)$. Therefore, \[
d_{SO}(n,k) \leq d(n,k)-2.
\]
Recently, Li et al.~\cite{LSHL-2026} showed that self-orthogonal codes with certain parameters do not exist.
\begin{enumerate}
    \item There are no binary self-orthogonal $[2^k-3, k, 2^{k-1}-2]$ and $[2^{k-1}-2, k, 2^{k-2}-2]$ codes for any $k \geq 4$.
    \item There are no binary self-orthogonal $[2^k-2^u-2, k, 2^{k-1}-2^{u-1}-2]$ codes for any $k \geq u+2 \geq 5$.
    \item There are no binary self-orthogonal $[2^k-1, k+1, 2^{k-1}-2]$ codes for $k \geq 5$.
\end{enumerate}
For each pair $(n,k)$, let $U_g(n,k)$ denote the upper bound on $d(n,k)$ reported in Grassl's table~\cite{G-online}. Since every binary self-orthogonal code is even, its minimum distance is even. We therefore initialize the candidate upper bound by
\[
U^{(0)}(n,k)
=
2\left\lfloor \frac{U_g(n,k)}{2}\right\rfloor.
\]
Starting with $d=U^{(0)}(n,k)$, we test each even candidate value in decreasing order. By Lemma~\ref{shi-res}, the existence of a binary self-orthogonal $[n,k,d]$ code would imply the existence of an even residual code with parameters
\[
\left[n-d, k-1, d' \ge 2 \left\lceil\frac{d}{4}\right\rceil\right].
\]
On the other hand, every even binary $[n-d,k-1]$ code has even minimum distance and hence has minimum distance at most $U_0(n-d,k-1)$. Consequently, if
\[
2 \left\lceil\frac{d}{4}\right\rceil>U_{0}(n-d,k-1),
\]
then the required residual code cannot exist, and therefore no binary self-orthogonal $[n,k,d]$ code exists. We also discard $d$ whenever the parameter $[n,k,d]$ belongs to one of the nonexistence families established by Li et al.~\cite{LSHL-2026}. Whenever a candidate is excluded by either criterion, we replace $d$ by $d-2$ and repeat the tests. The first candidate not excluded by these tests is denoted by $\overline{d}_{\mathrm{SO}}(n,k)$. This procedure is applied to all paramaters satisfying $22\le n\le 256$, $9\le k\le 15$, and $n\ge 2k$. The resulting bounds are listed in Table~\ref{tab:full-so-upper-bounds}. For $k=9$ and $k=10$, the values of $d_{SO}(n,k)$ are already known for $22\leq n\leq 40$ in~\cite{BBGO-2006}. Therefore, in these cases, we start from $n=41$.

\begingroup
\footnotesize
\begin{longtable}{c|ccccccc||c|ccccccc}
\caption{Upper bounds on the minimum distances of binary self-orthogonal codes
for $9\leq k\leq15$ and $22\leq n\leq256$.}
\label{tab:full-so-upper-bounds}\\
\hline
$n\backslash k$
& 9 & 10 & 11 & 12 & 13 & 14 & 15
& $n\backslash k$
& 9 & 10 & 11 & 12 & 13 & 14 & 15 \\
\hline
\endfirsthead
\hline
$n\backslash k$
& 9 & 10 & 11 & 12 & 13 & 14 & 15
& $n\backslash k$
& 9 & 10 & 11 & 12 & 13 & 14 & 15 \\
\hline
\endhead
22 &  &  & 6 &  &  &  &  & 140 & 66 & 64 & 64 & 64 & 64 & 62 & 62 \\
23 &  &  & 8 &  &  &  &  & 141 & 66 & 66 & 64 & 64 & 64 & 64 & 62 \\
24 &  &  & 8 & 8 &  &  &  & 142 & 68 & 66 & 66 & 64 & 64 & 64 & 64 \\
25 &  &  & 8 & 8 &  &  &  & 143 & 68 & 68 & 66 & 64 & 64 & 64 & 64 \\
26 &  &  & 8 & 8 & 6 &  &  & 144 & 68 & 68 & 68 & 66 & 64 & 64 & 64 \\
27 &  &  & 8 & 8 & 8 &  &  & 145 & 68 & 68 & 68 & 66 & 66 & 64 & 64 \\
28 &  &  & 8 & 8 & 8 & 8 &  & 146 & 68 & 68 & 68 & 68 & 66 & 66 & 64 \\
29 &  &  & 8 & 8 & 8 & 8 &  & 147 & 70 & 68 & 68 & 68 & 68 & 66 & 66 \\
30 &  &  & 10 & 8 & 8 & 8 & 8 & 148 & 70 & 70 & 68 & 68 & 68 & 66 & 66 \\
31 &  &  & 10 & 10 & 8 & 8 & 8 & 149 & 72 & 70 & 70 & 68 & 68 & 68 & 66 \\
32 &  &  & 12 & 10 & 10 & 8 & 8 & 150 & 72 & 72 & 70 & 68 & 68 & 68 & 68 \\
33 &  &  & 12 & 10 & 10 & 10 & 8 & 151 & 72 & 72 & 70 & 70 & 68 & 68 & 68 \\
34 &  &  & 12 & 12 & 10 & 10 & 10 & 152 & 72 & 72 & 72 & 70 & 70 & 68 & 68 \\
35 &  &  & 12 & 12 & 10 & 10 & 10 & 153 & 72 & 72 & 72 & 70 & 70 & 70 & 68 \\
36 &  &  & 12 & 12 & 12 & 10 & 10 & 154 & 72 & 72 & 72 & 72 & 70 & 70 & 70 \\
37 &  &  & 14 & 12 & 12 & 12 & 10 & 155 & 74 & 72 & 72 & 72 & 70 & 70 & 70 \\
38 &  &  & 14 & 14 & 12 & 12 & 12 & 156 & 74 & 74 & 72 & 72 & 72 & 70 & 70 \\
39 &  &  & 14 & 14 & 12 & 12 & 12 & 157 & 76 & 74 & 72 & 72 & 72 & 72 & 70 \\
40 &  &  & 14 & 14 & 14 & 12 & 12 & 158 & 76 & 76 & 74 & 72 & 72 & 72 & 72 \\
41 & 16 & 16 & 16 & 14 & 14 & 14 & 12 & 159 & 76 & 76 & 74 & 74 & 72 & 72 & 72 \\
42 & 16 & 16 & 16 & 16 & 14 & 14 & 14 & 160 & 76 & 76 & 76 & 74 & 74 & 72 & 72 \\
43 & 16 & 16 & 16 & 16 & 14 & 14 & 14 & 161 & 76 & 76 & 76 & 74 & 74 & 72 & 72 \\
44 & 18 & 16 & 16 & 16 & 16 & 14 & 14 & 162 & 78 & 76 & 76 & 76 & 74 & 74 & 72 \\
45 & 18 & 18 & 16 & 16 & 16 & 16 & 14 & 163 & 78 & 78 & 76 & 76 & 74 & 74 & 74 \\
46 & 20 & 18 & 18 & 16 & 16 & 16 & 16 & 164 & 80 & 78 & 76 & 76 & 76 & 74 & 74 \\
47 & 20 & 20 & 18 & 16 & 16 & 16 & 16 & 165 & 80 & 78 & 78 & 76 & 76 & 76 & 74 \\
48 & 20 & 20 & 18 & 18 & 16 & 16 & 16 & 166 & 80 & 80 & 78 & 76 & 76 & 76 & 76 \\
49 & 20 & 20 & 20 & 18 & 18 & 16 & 16 & 167 & 80 & 80 & 78 & 78 & 76 & 76 & 76 \\
50 & 20 & 20 & 20 & 20 & 18 & 18 & 16 & 168 & 80 & 80 & 80 & 78 & 78 & 76 & 76 \\
51 & 22 & 20 & 20 & 20 & 20 & 18 & 18 & 169 & 80 & 80 & 80 & 78 & 78 & 76 & 76 \\
52 & 22 & 22 & 20 & 20 & 20 & 18 & 18 & 170 & 80 & 80 & 80 & 80 & 78 & 78 & 76 \\
53 & 22 & 22 & 22 & 20 & 20 & 20 & 18 & 171 & 82 & 80 & 80 & 80 & 78 & 78 & 78 \\
54 & 24 & 22 & 22 & 22 & 20 & 20 & 20 & 172 & 82 & 82 & 80 & 80 & 80 & 78 & 78 \\
55 & 24 & 24 & 22 & 22 & 20 & 20 & 20 & 173 & 84 & 82 & 80 & 80 & 80 & 80 & 78 \\
56 & 24 & 24 & 24 & 22 & 22 & 20 & 20 & 174 & 84 & 84 & 82 & 80 & 80 & 80 & 80 \\
57 & 24 & 24 & 24 & 22 & 22 & 20 & 20 & 175 & 84 & 84 & 82 & 82 & 80 & 80 & 80 \\
58 & 24 & 24 & 24 & 24 & 22 & 22 & 20 & 176 & 84 & 84 & 84 & 82 & 80 & 80 & 80 \\
59 & 26 & 24 & 24 & 24 & 24 & 22 & 22 & 177 & 84 & 84 & 84 & 82 & 82 & 80 & 80 \\
60 & 26 & 24 & 24 & 24 & 24 & 22 & 22 & 178 & 86 & 84 & 84 & 84 & 82 & 82 & 80 \\
61 & 26 & 26 & 24 & 24 & 24 & 24 & 22 & 179 & 86 & 86 & 84 & 84 & 82 & 82 & 82 \\
62 & 28 & 26 & 26 & 24 & 24 & 24 & 24 & 180 & 88 & 86 & 84 & 84 & 84 & 82 & 82 \\
63 & 28 & 28 & 26 & 26 & 24 & 24 & 24 & 181 & 88 & 86 & 86 & 84 & 84 & 84 & 82 \\
64 & 28 & 28 & 26 & 26 & 26 & 24 & 24 & 182 & 88 & 88 & 86 & 84 & 84 & 84 & 84 \\
65 & 28 & 28 & 28 & 26 & 26 & 24 & 24 & 183 & 88 & 88 & 86 & 86 & 84 & 84 & 84 \\
66 & 30 & 28 & 28 & 28 & 26 & 26 & 24 & 184 & 88 & 88 & 88 & 86 & 84 & 84 & 84 \\
67 & 30 & 28 & 28 & 28 & 28 & 26 & 26 & 185 & 88 & 88 & 88 & 86 & 86 & 84 & 84 \\
68 & 30 & 30 & 28 & 28 & 28 & 28 & 26 & 186 & 90 & 88 & 88 & 88 & 86 & 86 & 84 \\
69 & 30 & 30 & 28 & 28 & 28 & 28 & 26 & 187 & 90 & 88 & 88 & 88 & 86 & 86 & 86 \\
70 & 32 & 30 & 30 & 28 & 28 & 28 & 28 & 188 & 90 & 90 & 88 & 88 & 88 & 86 & 86 \\
71 & 32 & 32 & 30 & 30 & 28 & 28 & 28 & 189 & 92 & 90 & 88 & 88 & 88 & 86 & 86 \\
72 & 32 & 32 & 32 & 30 & 30 & 28 & 28 & 190 & 92 & 92 & 90 & 88 & 88 & 88 & 86 \\
73 & 32 & 32 & 32 & 30 & 30 & 28 & 28 & 191 & 92 & 92 & 90 & 90 & 88 & 88 & 88 \\
74 & 32 & 32 & 32 & 32 & 30 & 30 & 28 & 192 & 92 & 92 & 92 & 90 & 88 & 88 & 88 \\
75 & 32 & 32 & 32 & 32 & 32 & 30 & 30 & 193 & 94 & 92 & 92 & 90 & 90 & 88 & 88 \\
76 & 34 & 32 & 32 & 32 & 32 & 32 & 30 & 194 & 94 & 92 & 92 & 92 & 90 & 90 & 88 \\
77 & 34 & 34 & 32 & 32 & 32 & 32 & 30 & 195 & 94 & 94 & 92 & 92 & 90 & 90 & 90 \\
78 & 36 & 34 & 34 & 32 & 32 & 32 & 32 & 196 & 96 & 94 & 92 & 92 & 90 & 90 & 90 \\
79 & 36 & 36 & 34 & 34 & 32 & 32 & 32 & 197 & 96 & 94 & 94 & 92 & 92 & 90 & 90 \\
80 & 36 & 36 & 34 & 34 & 32 & 32 & 32 & 198 & 96 & 96 & 94 & 92 & 92 & 92 & 90 \\
81 & 36 & 36 & 36 & 34 & 34 & 32 & 32 & 199 & 96 & 96 & 96 & 94 & 92 & 92 & 92 \\
82 & 36 & 36 & 36 & 36 & 34 & 34 & 32 & 200 & 96 & 96 & 96 & 94 & 92 & 92 & 92 \\
83 & 38 & 36 & 36 & 36 & 36 & 34 & 34 & 201 & 96 & 96 & 96 & 94 & 94 & 92 & 92 \\
84 & 38 & 38 & 36 & 36 & 36 & 36 & 34 & 202 & 96 & 96 & 96 & 96 & 94 & 94 & 92 \\
85 & 40 & 38 & 38 & 36 & 36 & 36 & 34 & 203 & 98 & 96 & 96 & 96 & 94 & 94 & 94 \\
86 & 40 & 38 & 38 & 36 & 36 & 36 & 36 & 204 & 98 & 96 & 96 & 96 & 96 & 94 & 94 \\
87 & 40 & 40 & 38 & 38 & 36 & 36 & 36 & 205 & 100 & 98 & 96 & 96 & 96 & 94 & 94 \\
88 & 40 & 40 & 38 & 38 & 38 & 36 & 36 & 206 & 100 & 98 & 98 & 96 & 96 & 96 & 94 \\
89 & 40 & 40 & 40 & 38 & 38 & 38 & 36 & 207 & 100 & 100 & 98 & 98 & 96 & 96 & 96 \\
90 & 42 & 40 & 40 & 40 & 38 & 38 & 36 & 208 & 100 & 100 & 100 & 98 & 96 & 96 & 96 \\
91 & 42 & 40 & 40 & 40 & 40 & 38 & 38 & 209 & 100 & 100 & 100 & 100 & 98 & 96 & 96 \\
92 & 42 & 40 & 40 & 40 & 40 & 40 & 38 & 210 & 102 & 100 & 100 & 100 & 98 & 98 & 96 \\
93 & 44 & 42 & 40 & 40 & 40 & 40 & 38 & 211 & 102 & 100 & 100 & 100 & 98 & 98 & 98 \\
94 & 44 & 42 & 42 & 40 & 40 & 40 & 40 & 212 & 104 & 102 & 100 & 100 & 100 & 98 & 98 \\
95 & 44 & 44 & 42 & 42 & 40 & 40 & 40 & 213 & 104 & 102 & 102 & 100 & 100 & 100 & 98 \\
96 & 44 & 44 & 42 & 42 & 42 & 40 & 40 & 214 & 104 & 104 & 102 & 100 & 100 & 100 & 100 \\
97 & 44 & 44 & 44 & 42 & 42 & 40 & 40 & 215 & 104 & 104 & 104 & 102 & 100 & 100 & 100 \\
98 & 46 & 44 & 44 & 44 & 42 & 42 & 40 & 216 & 104 & 104 & 104 & 102 & 102 & 100 & 100 \\
99 & 46 & 44 & 44 & 44 & 44 & 42 & 42 & 217 & 104 & 104 & 104 & 102 & 102 & 100 & 100 \\
100 & 46 & 46 & 44 & 44 & 44 & 42 & 42 & 218 & 106 & 104 & 104 & 104 & 102 & 102 & 100 \\
101 & 48 & 46 & 46 & 44 & 44 & 44 & 42 & 219 & 106 & 104 & 104 & 104 & 102 & 102 & 102 \\
102 & 48 & 46 & 46 & 46 & 44 & 44 & 44 & 220 & 108 & 106 & 104 & 104 & 102 & 102 & 102 \\
103 & 48 & 48 & 46 & 46 & 44 & 44 & 44 & 221 & 108 & 106 & 106 & 104 & 104 & 102 & 102 \\
104 & 48 & 48 & 46 & 46 & 46 & 44 & 44 & 222 & 108 & 108 & 106 & 104 & 104 & 104 & 102 \\
105 & 48 & 48 & 48 & 46 & 46 & 46 & 44 & 223 & 108 & 108 & 108 & 106 & 104 & 104 & 104 \\
106 & 48 & 48 & 48 & 48 & 46 & 46 & 44 & 224 & 108 & 108 & 108 & 106 & 104 & 104 & 104 \\
107 & 48 & 48 & 48 & 48 & 48 & 46 & 46 & 225 & 110 & 108 & 108 & 106 & 106 & 104 & 104 \\
108 & 50 & 48 & 48 & 48 & 48 & 46 & 46 & 226 & 110 & 108 & 108 & 108 & 106 & 106 & 104 \\
109 & 50 & 50 & 48 & 48 & 48 & 48 & 46 & 227 & 110 & 110 & 108 & 108 & 106 & 106 & 106 \\
110 & 52 & 50 & 50 & 48 & 48 & 48 & 48 & 228 & 112 & 110 & 110 & 108 & 106 & 106 & 106 \\
111 & 52 & 52 & 50 & 48 & 48 & 48 & 48 & 229 & 112 & 112 & 110 & 108 & 108 & 106 & 106 \\
112 & 52 & 52 & 50 & 50 & 48 & 48 & 48 & 230 & 112 & 112 & 110 & 110 & 108 & 108 & 106 \\
113 & 52 & 52 & 52 & 50 & 50 & 48 & 48 & 231 & 112 & 112 & 112 & 110 & 108 & 108 & 108 \\
114 & 52 & 52 & 52 & 52 & 50 & 50 & 48 & 232 & 112 & 112 & 112 & 110 & 108 & 108 & 108 \\
115 & 54 & 52 & 52 & 52 & 52 & 50 & 50 & 233 & 112 & 112 & 112 & 110 & 110 & 108 & 108 \\
116 & 54 & 54 & 52 & 52 & 52 & 50 & 50 & 234 & 114 & 112 & 112 & 112 & 110 & 110 & 108 \\
117 & 56 & 54 & 54 & 52 & 52 & 52 & 50 & 235 & 114 & 112 & 112 & 112 & 112 & 110 & 110 \\
118 & 56 & 56 & 54 & 52 & 52 & 52 & 52 & 236 & 116 & 114 & 112 & 112 & 112 & 112 & 110 \\
119 & 56 & 56 & 54 & 54 & 52 & 52 & 52 & 237 & 116 & 114 & 114 & 112 & 112 & 112 & 112 \\
120 & 56 & 56 & 56 & 54 & 54 & 52 & 52 & 238 & 116 & 116 & 114 & 112 & 112 & 112 & 112 \\
121 & 56 & 56 & 56 & 54 & 54 & 54 & 52 & 239 & 116 & 116 & 116 & 114 & 112 & 112 & 112 \\
122 & 56 & 56 & 56 & 56 & 54 & 54 & 54 & 240 & 116 & 116 & 116 & 114 & 112 & 112 & 112 \\
123 & 58 & 56 & 56 & 56 & 56 & 54 & 54 & 241 & 116 & 116 & 116 & 116 & 114 & 112 & 112 \\
124 & 58 & 58 & 56 & 56 & 56 & 54 & 54 & 242 & 118 & 116 & 116 & 116 & 114 & 114 & 112 \\
125 & 60 & 58 & 56 & 56 & 56 & 56 & 54 & 243 & 118 & 118 & 116 & 116 & 116 & 114 & 114 \\
126 & 60 & 58 & 58 & 56 & 56 & 56 & 56 & 244 & 120 & 118 & 116 & 116 & 116 & 114 & 114 \\
127 & 60 & 60 & 58 & 58 & 56 & 56 & 56 & 245 & 120 & 120 & 118 & 116 & 116 & 114 & 114 \\
128 & 60 & 60 & 60 & 58 & 58 & 56 & 56 & 246 & 120 & 120 & 118 & 118 & 116 & 116 & 114 \\
129 & 60 & 60 & 60 & 58 & 58 & 56 & 56 & 247 & 120 & 120 & 120 & 118 & 116 & 116 & 116 \\
130 & 62 & 60 & 60 & 60 & 58 & 58 & 56 & 248 & 122 & 120 & 120 & 120 & 118 & 116 & 116 \\
131 & 62 & 60 & 60 & 60 & 60 & 58 & 58 & 249 & 122 & 120 & 120 & 120 & 118 & 118 & 116 \\
132 & 62 & 62 & 60 & 60 & 60 & 58 & 58 & 250 & 122 & 120 & 120 & 120 & 120 & 118 & 116 \\
133 & 64 & 62 & 62 & 60 & 60 & 60 & 58 & 251 & 122 & 122 & 120 & 120 & 120 & 120 & 118 \\
134 & 64 & 64 & 62 & 60 & 60 & 60 & 60 & 252 & 124 & 122 & 120 & 120 & 120 & 120 & 118 \\
135 & 64 & 64 & 62 & 62 & 60 & 60 & 60 & 253 & 124 & 124 & 122 & 120 & 120 & 120 & 120 \\
136 & 64 & 64 & 64 & 62 & 62 & 60 & 60 & 254 & 124 & 124 & 122 & 120 & 120 & 120 & 120 \\
137 & 64 & 64 & 64 & 62 & 62 & 62 & 60 & 255 & 124 & 124 & 124 & 122 & 120 & 120 & 120 \\
138 & 64 & 64 & 64 & 64 & 62 & 62 & 62 & 256 & 128 & 124 & 124 & 122 & 120 & 120 & 120 \\
139 & 64 & 64 & 64 & 64 & 64 & 62 & 62 &  &  &  &  &  &  &  &  \\
\hline
\end{longtable}
\endgroup

We next describe our genetic algorithm for constructing shortest self-orthogonal embeddings with large minimum distance. Given a binary linear code $\mathcal{C}$, our goal is to find, among all shortest self-orthogonal embeddings of $\mathcal{C}$, one with the largest possible minimum distance.
\begin{enumerate}
    \item We aim to obtain self-orthogonal codes with large minimum distance of parameters in the range
    \[
    22\leq n\leq 256
    \qquad\text{and}\qquad
    9\leq k\leq 15.
    \]
    We use base codes from the MAGMA BKLC (Best Known Linear Codes) library and apply shortest self-orthogonal embedding methods.
    \item Let $\mathcal{C}$ be an $[n,k,d]$ code over $\mathbb{F}_2$ with generator matrix $G$, obtained from the MAGMA BKLC database, and assume that $\mathcal{C}$ is not self-orthogonal. Let $\ell=\dim\Hull(\mathcal{C})$. By Theorem~\ref{AKKLW-2025}, the number $m$ of appended columns in a shortest self-orthogonal embedding is $k-\ell$ if $\mathcal{C}$ is an odd-like code and $k-\ell+1$ if $\mathcal{C}$ is an even code.
    \item In our genetic algorithm, we use
    \[
    \mathcal{T}=\{T_{\mathbf{u}}\mid \wt(\mathbf{u})=4,\ \mathbf{u}\in\mathbb{F}_2^m\}
    \]
    as the set of genes. However, by Theorem~\ref{orthogonal-group-generator}, $\mathcal{T}$ generates the orthogonal group only when $m\geq5$. For $m\leq3$, the orthogonal group $\mathcal{O}_m$ coincides with the group $\mathcal{P}_m$ of $m\times m$ permutation matrices, while for $m=4$, $\mathcal{O}_4$ is not generated solely by $\mathcal{T}$. In these cases, since the orthogonal group is sufficiently small, we examine all elements of $\mathcal{O}_m$ exhaustively. Hence, we apply the genetic algorithm only for $m\geq5$.
    \item \textbf{Exhaustive Search Algorithm.} When the number of added columns $m$ satisfies $m \le 4$, we act the orthogonal group of degree $m$ to the shortest self-orthogonal embedding constructed in Remark~\ref{base_SO_emb_alg}, and select, among the resulting shortest self-orthogonal embeddings, one with the largest minimum distance.
    \item \textbf{Genetic Algorithm.} The initial feasible matrix $B_0$ is constructed using the algorithm given in Remark~\ref{base_SO_emb_alg}. Once the initial matrix $B_0$ is found, the genetic algorithm is run for at most $120$ generations, and the initial random chromosome length is bounded above by $3m$.
    \begin{enumerate}
        \item Initialization: The initial population consists of $120$ randomly generated chromosomes. Within the upper bound of chromosome length at most $3m$, the length of a chromosome is randomly selected. Genes are selected from $\mathcal{T}$ for the predetermined chromosome length.

        \item Selection: Parent selection is performed by tournament selection with tournament size $5$. The top $2$ chromosomes are carried over from the previous generation as elites and $59$ pairs of parents are selected by tournament selection. In the selection phase, we use the fitness function $F(\chi)=\left(d_{\chi},-A_{d_{\chi}}(\mathcal{C}_{\chi}),-|\chi|\right)$ to evaluate each chromosome.

        \item Crossover: The crossover probability is set to $0.8$. If crossover is not applied, which occurs with probability $0.2$, the two parent chromosomes are used as the two offspring. When crossover is applied, random crossover is chosen with probability $0.45$, and guided crossover is chosen with probability $0.55$. Regarding guided crossover, we select $p=\min\{4,s\}$ prefix chromosomes from the base chromosome and $q=\min\{3,s'\}$ suffix chromosomes from the secondary chromosome where $s$ and $s'$ are the lengths of primary and secondary chromosomes, respectively. We choose two offspring chromosomes which satisfy the length constraint $|\chi|\leq 4m$ and have the largest fitness values. If fewer than two feasible candidate chromosomes are obtained, we instead perform a random crossover on the same pair of parents.

        \item Mutation: The mutation probability is set to $0.45$. When mutation is applied, the relative weights of the mutation operations are given by $0.29$ for replace, $0.29$ for insert, $0.21$ for delete, and $0.21$ for swap. Moreover, the maximum allowable chromosome length is bounded above by $4m$.
    \end{enumerate}
    If the best fitness value does not improve for $20$ consecutive generations, a restart is performed by replacing $50\%$ of the population with newly generated chromosomes. For each given BKLC code, the self-orthogonal embedding with the highest fitness value among the shortest embeddings found by the genetic algorithm is retained.
\end{enumerate}
\begin{algorithm}[t]
\caption{Genetic Algorithm for Shortest Self-Orthogonal Embedding}
\label{alg:ga-so}
\begin{algorithmic}[1]
\Require A BKLC code $\mathcal{C}$ requiring $m\ (\ge 5)$ appended columns
\Ensure A self-orthogonal embedding of $\mathcal{C}$ with large minimum distance
\State Split the generator matrix of $\mathcal{C}$ into its hull part and complementary part $A$
\State Find an initial feasible matrix $B_0$ satisfying $AA^T+B_0B_0^T=\mathcal{O}$ using algorithm in Remark~\ref{base_SO_emb_alg}
\State Set the gene pool $\mathcal{T}=\{\,T_\mathbf{u}: \wt(\mathbf{u})=4,\ \mathbf{u}\in\mathbb{F}_2^m\,\}$
    \Statex \Comment{each chromosome is a gene sequence and yields an embedding by applying its genes to $B_0$}
\State Initialize $120$ chromosomes with random genes from $\mathcal{T}$, each of length at most $3m$
\For{at most $120$ generations}
    \State Carry over the best $2$ chromosomes as elites
    \For{$59$ parent pairs selected by tournament selection (size $5$)}
        \State With probability $0.8$ apply crossover (guided $0.55$, random $0.45$)
        \State With probability $0.45$, mutate each offspring by an
        admissible operator according to the chromosome length
        \Statex \Comment{delete requires $|\chi|\geq2$, swap requires
        $|\chi|\geq2$, and insert requires $|\chi|<4m$}
        \State Add the two offspring, respecting the length bound $4m$
    \EndFor
    \State Rank chromosomes by the fitness function $F(\chi)=(d_{\chi},-A_{d_{\chi}}(\mathcal{C}_{\chi}),-|\chi|)$ in lexicographical order
    \If{the best value has not improved for $20$ generations}
        \State Replace $50\%$ of the population with new random chromosomes
    \EndIf
\EndFor
\State \Return the best embedding found
\end{algorithmic}
\end{algorithm}

\section{Numerical Results}
We conducted experiments on 1026 codes from the MAGMA BKLC database. Among these codes, 557 codes with $m \le 4$ were examined using exhaustive search, while the genetic algorithm was applied to the remaining 469 codes with $m \geq 5$, where $m$ denotes the number of columns appended in a shortest self-orthogonal embedding. For each base code, our algorithm searched for an embedding that maximizes the minimum distance of the resulting self-orthogonal code. The results obtained from these experiments are summarized in Table~\ref{alg-result-sum}. In Table~\ref{alg-result-sum}, \textit{Distance from bound} denotes the difference between the upper bound given in Table~\ref{tab:full-so-upper-bounds} and the minimum distance of a code with the corresponding parameters obtained in our experiments. \textit{Inequivalent codes} denotes the number of equivalence classes obtained by classifying all codes produced by the genetic algorithm and exhaustive search up to equivalence.
\begin{table}[ht]
\centering
\caption{Results of the exhaustive search and genetic algorithm.}
\label{alg-result-sum}
\begin{tabular}{c|r|r|r}
\hline
Distance from bound & Exhaustive search & Genetic algorithm & Inequivalent codes\\
\hline
0  & 350 & 73  & 280\\
2  & 170 & 147 & 273\\
4  & 33  & 178 & 208\\
6  & 4   & 52 & 56\\
8  & 0   & 18 & 18\\
10 & 0   & 1 & 1\\
\hline
Total & 557 & 469 & 836 \\
\hline
\end{tabular}
\end{table}

We compared the 280 inequivalent self-orthogonal codes attaining the bound with the codes provided in~\cite{B-1999,G-online,BBK-online,KY-2014,KW-2009}. Among them, \(146\) codes were equivalent to previously known codes. Another \(100\) codes had the same parameters as previously known codes but were inequivalent to them. For the remaining \(34\) codes, no codes with the corresponding parameters were previously known. Hence, these codes have new parameters.

The generator matrices of the 34 codes which have new parameters are presented in hexadecimal format $(A=10, B=11, C=12, D=13, E=14, F=15)$ after appending them with zero columns on the right to ensure that the length is a multiple of 4.
\begingroup
\setlength{\arraycolsep}{0pt}
\begin{itemize}
\item A binary self-orthogonal $[53, 9, 22]$ code which has an automorphism group of size $|\mathrm{Aut}|=1$ and a generator matrix
{\scriptsize
\setlength{\arraycolsep}{0.35pt}%
\renewcommand{\arraystretch}{0.9}%
\[
G_{\mathrm{SO},12}=
\left[
\begin{array}{@{}*{14}{c}@{}}
9 & 1 & 7 & 0 & D & F & 5 & C & F & A & 4 & B & 4 & 0 \\
5 & 0 & 5 & 9 & F & 0 & 7 & B & 4 & A & D & 1 & 8 & 0 \\
3 & 0 & C & 5 & 5 & 7 & A & 9 & D & 4 & 9 & 9 & 4 & 0 \\
0 & 8 & 8 & A & E & C & 4 & 5 & 5 & 7 & 7 & D & 8 & 0 \\
0 & 5 & 5 & F & A & 2 & A & F & F & D & 0 & 0 & 0 & 0 \\
0 & 2 & 2 & 5 & B & 9 & 1 & 2 & A & D & F & D & 8 & 0 \\
8 & 0 & D & D & 3 & 8 & 6 & 0 & 9 & 4 & 7 & 3 & F & 0 \\
4 & 0 & E & 7 & 6 & 8 & 3 & 1 & 4 & 0 & E & D & 5 & 0 \\
2 & 0 & 7 & B & C & 8 & 1 & C & 2 & 1 & 5 & B & 8 & 8
\end{array}\right].
\]
}
\item A binary self-orthogonal $[61, 9, 26]$ code which has an automorphism group of size $|\mathrm{Aut}|=2$ and a generator matrix
{\scriptsize
\setlength{\arraycolsep}{0.35pt}%
\renewcommand{\arraystretch}{0.9}%
\[
G_{\mathrm{SO},13}=
\left[\begin{array}{@{}*{16}{c}@{}}
8 & 2 & 2 & 8 & 8 & E & F & F & 5 & 6 & 4 & D & 5 & C & 0 & 0 \\
4 & 2 & 1 & 4 & B & 1 & 3 & C & 3 & 5 & 2 & D & E & D & C & 0 \\
2 & 8 & 7 & B & 2 & B & 0 & 9 & 1 & E & 3 & 8 & 7 & D & 8 & 0 \\
1 & 1 & 1 & 1 & 1 & 7 & F & B & 9 & 9 & A & 6 & 7 & 8 & 8 & 0 \\
0 & 6 & 5 & 6 & F & C & C & 1 & 9 & B & 7 & F & 5 & 4 & C & 0 \\
0 & 0 & F & F & 0 & 0 & F & 7 & F & F & 8 & 1 & E & 7 & 8 & 0 \\
8 & 2 & 2 & 7 & 7 & 2 & F & 2 & 1 & F & F & 8 & 7 & 1 & 2 & 0 \\
2 & 1 & 4 & 7 & 2 & 4 & 5 & 7 & 0 & 2 & F & 2 & D & E & 9 & 8 \\
0 & 6 & 5 & 6 & F & 3 & C & 6 & 6 & 4 & 8 & 0 & A & B & 0 & 8
\end{array}\right].
\]
}
\item A binary self-orthogonal $[69, 9, 30]$ code which has an automorphism group of size $|\mathrm{Aut}|=1512$ and a generator matrix
{\scriptsize
\setlength{\arraycolsep}{0.35pt}%
\renewcommand{\arraystretch}{0.9}%
\[
G_{\mathrm{SO},14}=
\left[\begin{array}{@{}*{18}{c}@{}}
8 & 2 & F & C & A & 8 & C & F & 7 & 5 & A & 6 & C & 4 & 8 & 6 & 0 & 0 \\
4 & 3 & 8 & 2 & F & C & A & 8 & C & F & 7 & 5 & A & 6 & C & 4 & 0 & 0 \\
2 & 1 & C & 1 & 7 & E & 5 & 4 & 6 & 7 & B & A & D & 3 & 6 & 2 & 0 & 0 \\
1 & 2 & 1 & C & 1 & 7 & E & 5 & 4 & 6 & 7 & B & A & D & 3 & 6 & 0 & 0 \\
0 & B & F & 2 & A & 3 & 3 & D & D & 6 & 9 & B & 1 & 2 & 1 & C & 0 & 0 \\
0 & 5 & F & 9 & 5 & 1 & 9 & E & E & B & 4 & D & 8 & 9 & 0 & E & 0 & 0 \\
8 & 0 & 3 & D & A & D & C & 0 & 9 & 3 & 3 & A & 6 & 3 & 2 & 3 & A & 0 \\
4 & 0 & 2 & 3 & 7 & B & 2 & 0 & D & A & A & 7 & 5 & 2 & B & 2 & E & 8 \\
2 & 0 & 5 & 7 & 4 & B & D & 1 & D & 8 & 1 & D & 0 & C & 3 & D & 9 & 8
\end{array}\right].
\]
}
\item A binary self-orthogonal $[84, 9, 38]$ code which has an automorphism group of size $|\mathrm{Aut}|=4$ and a generator matrix
{\scriptsize
\setlength{\arraycolsep}{0.35pt}%
\renewcommand{\arraystretch}{0.9}%
\[
G_{\mathrm{SO},15}=
\left[\begin{array}{@{}*{21}{c}@{}}
8 & 0 & 3 & 6 & 7 & 5 & F & 6 & C & B & 6 & B & 6 & 0 & 5 & E & 0 & 8 & A & 9 & 0 \\
4 & 0 & 3 & 5 & C & E & 3 & 5 & C & A & F & A & 0 & D & 2 & 5 & 6 & B & F & 2 & 0 \\
2 & 2 & D & 5 & C & 0 & D & 8 & 4 & 0 & 6 & F & D & 2 & A & 3 & 4 & 9 & E & 7 & 8 \\
1 & 6 & 8 & 0 & E & D & 1 & 8 & 8 & 7 & 7 & D & 6 & D & 9 & 6 & C & 3 & 2 & B & 8 \\
0 & E & 8 & 0 & 7 & 9 & 6 & E & 1 & 7 & F & 2 & D & 1 & 7 & 3 & 7 & 0 & 8 & 4 & 8 \\
0 & 1 & D & 5 & 3 & 7 & 5 & 6 & 3 & C & A & 4 & 8 & B & 0 & D & 6 & 8 & 9 & B & 8 \\
0 & 0 & 0 & E & 8 & F & 4 & E & 4 & 6 & 3 & C & B & 7 & A & B & B & 5 & D & 1 & 8 \\
2 & 0 & 3 & 3 & 3 & 7 & 9 & 5 & 6 & C & 3 & 2 & 0 & 2 & B & B & 9 & 3 & 6 & 2 & D \\
1 & 4 & 0 & 3 & 7 & A & 3 & 3 & C & 0 & 3 & E & E & 6 & 0 & 3 & E & C & 8 & A & B
\end{array}\right].
\]
}
\item A binary self-orthogonal $[109, 9, 50]$ code which has an automorphism group of size $|\mathrm{Aut}|=3$ and a generator matrix
{\scriptsize
\setlength{\arraycolsep}{0.35pt}%
\renewcommand{\arraystretch}{0.9}%
\[
G_{\mathrm{SO},16}=
\left[\begin{array}{@{}*{28}{c}@{}}
8 & 5 & E & D & 7 & 2 & 4 & 9 & 2 & F & 6 & D & D & 7 & 4 & 8 & A & 7 & 9 & 3 & D & A & E & C & E & F & 0 & 0 \\
4 & 5 & 1 & 1 & 4 & A & A & A & D & 2 & A & B & C & 9 & 3 & 9 & 6 & 4 & 6 & 3 & D & 9 & 0 & B & 3 & F & C & 0 \\
2 & 5 & 8 & 4 & 6 & 4 & 4 & 8 & 8 & B & 9 & A & F & A & B & A & D & 9 & F & 0 & F & 5 & D & 0 & 6 & 6 & 8 & 0 \\
1 & 4 & 5 & 2 & 8 & 9 & 4 & A & 4 & 5 & 6 & F & 9 & 5 & B & D & 9 & E & A & 5 & 1 & E & 7 & E & 1 & 3 & 0 & 0 \\
0 & D & 2 & 5 & 2 & A & 1 & E & B & 3 & 3 & 4 & 8 & E & F & 6 & 8 & 3 & 4 & 5 & 9 & 7 & B & 1 & B & 9 & 4 & 0 \\
0 & 3 & F & E & B & 2 & 6 & 6 & E & 4 & A & F & E & 0 & 1 & 6 & A & E & 3 & 8 & F & 0 & 0 & 6 & 8 & 3 & C & 0 \\
8 & 5 & E & 0 & E & E & B & 5 & C & 4 & 0 & 6 & 0 & A & 1 & 3 & 6 & F & 9 & 3 & B & 6 & 1 & 9 & 5 & 6 & 9 & 8 \\
4 & 5 & 1 & 0 & B & 5 & 5 & 5 & 2 & D & 5 & 3 & C & 9 & 4 & 6 & 9 & B & 9 & C & 2 & 7 & 0 & C & D & F & B & 8 \\
2 & 5 & 8 & 2 & E & 8 & 7 & 8 & 2 & 8 & 9 & 7 & 5 & 1 & 6 & A & 1 & E & 0 & F & 5 & C & B & D & 6 & 8 & 4 & 8
\end{array}\right].
\]
}
\item A binary self-orthogonal $[132, 9, 62]$ code which has an automorphism group of size $|\mathrm{Aut}|=288$ and a generator matrix
{\scriptsize
\setlength{\arraycolsep}{0.35pt}%
\renewcommand{\arraystretch}{0.9}%
\[
G_{\mathrm{SO},17}=
\left[\begin{array}{@{}*{33}{c}@{}}
9 & 1 & 2 & 5 & B & F & 0 & 5 & 4 & 8 & B & 5 & 6 & E & 4 & E & B & 2 & D & 6 & A & F & 3 & 1 & 6 & 6 & 8 & 2 & 3 & 5 & 7 & 4 & 8 \\
4 & 1 & 4 & 9 & 9 & A & A & A & 6 & D & C & 6 & 6 & 8 & E & 6 & 7 & 3 & 6 & 9 & 2 & B & 1 & F & 6 & 9 & 1 & 8 & C & E & 6 & E & 0 \\
3 & 1 & C & 7 & F & 4 & 7 & 0 & F & C & 3 & C & 3 & 3 & E & 1 & E & 3 & C & 7 & 1 & C & 0 & F & 0 & 0 & F & 8 & 1 & F & C & 0 & 8 \\
0 & 9 & 8 & 4 & E & C & 9 & B & 2 & 7 & A & 5 & 3 & 5 & 9 & 3 & D & 6 & 5 & 4 & 7 & 8 & 7 & 6 & C & 3 & 3 & 3 & D & 8 & 5 & C & 8 \\
0 & 4 & 7 & 8 & C & D & 6 & 8 & D & 9 & A & A & 3 & B & 9 & 4 & D & 6 & A & 5 & 8 & 6 & 7 & 5 & 3 & 3 & 3 & 3 & A & 7 & 5 & 3 & 0 \\
0 & 3 & F & F & C & 0 & 0 & F & F & F & F & 0 & 0 & F & F & 8 & 0 & 0 & 0 & F & F & F & F & C & 0 & 0 & 0 & 0 & 0 & 0 & F & F & 8 \\
9 & 1 & 2 & 5 & B & 8 & F & A & B & 7 & 4 & A & 9 & 1 & B & 1 & 4 & D & 2 & 9 & 4 & F & 3 & 1 & 6 & 6 & 8 & 2 & 3 & 5 & 7 & 4 & D \\
4 & 1 & 4 & 9 & 9 & A & A & A & 6 & D & C & 6 & 6 & 7 & 1 & 9 & 8 & C & 9 & 6 & D & 4 & E & 0 & 9 & 6 & E & 7 & 3 & 1 & 9 & 1 & C \\
0 & 4 & 6 & 7 & 3 & 2 & 9 & 7 & 3 & 9 & A & A & 3 & 4 & 6 & B & 2 & E & A & 5 & 8 & 6 & 7 & 5 & 3 & 3 & 3 & C & 5 & 8 & A & C & B
\end{array}\right].
\]
}
\item A binary self-orthogonal $[188, 9, 90]$ code which has an automorphism group of size $|\mathrm{Aut}|=1$ and a generator matrix
{\scriptsize
\setlength{\arraycolsep}{0.35pt}%
\renewcommand{\arraystretch}{0.9}%
\[
G_{\mathrm{SO},18}=
\left[\begin{array}{@{}*{47}{c}@{}}
A & 0 & 0 & 1 & 7 & 3 & 0 & 1 & D & D & A & 8 & 0 & F & 0 & 6 & 1 & B & 1 & 4 & 5 & 7 & 7 & 6 & 1 & 9 & B & C & F & 6 & B & 8 & 5 & 6 & D & E & D & 5 & E & C & D & 8 & 6 & 1 & E & D & 8 \\
6 & 0 & 5 & A & 1 & E & 1 & B & 0 & 6 & 3 & 5 & 7 & 4 & 6 & C & C & 3 & 1 & 8 & 5 & C & 3 & 5 & D & A & D & C & 3 & 0 & 1 & 6 & D & B & 4 & 6 & B & 6 & E & 7 & 5 & E & D & A & B & 6 & 0 \\
1 & 4 & 3 & D & 4 & 5 & 0 & D & B & 7 & 4 & 0 & D & 6 & 8 & 0 & 2 & B & A & C & 3 & C & C & 3 & A & 1 & D & C & D & B & 0 & E & D & B & B & 3 & 1 & B & 5 & 1 & 6 & D & 0 & D & E & B & 8 \\
0 & C & 5 & 7 & 8 & 3 & B & 4 & 0 & A & C & 5 & 6 & F & 0 & B & 5 & D & 7 & 6 & D & 6 & 0 & 0 & D & 4 & 6 & 6 & 2 & 8 & 1 & 7 & 6 & C & D & D & A & D & E & A & F & 3 & A & 0 & 3 & 7 & 8 \\
0 & 2 & E & 6 & 0 & 3 & B & B & 5 & 0 & 1 & E & 0 & C & 3 & 6 & 2 & 8 & A & E & E & C & 3 & 3 & 7 & 9 & E & D & 7 & 0 & A & D & B & D & A & B & D & 9 & B & 0 & C & 3 & D & B & 0 & A & 0 \\
0 & 1 & E & B & 5 & B & 6 & C & 6 & 7 & A & E & B & B & 0 & 7 & 4 & 6 & 7 & 5 & B & 1 & 6 & B & 1 & 7 & 5 & A & 2 & E & C & E & 8 & B & 4 & 6 & 0 & F & 6 & 7 & B & 0 & 0 & E & B & 1 & 8 \\
6 & 0 & 0 & 6 & D & E & 6 & C & 6 & C & 3 & 6 & B & 5 & E & A & 0 & 6 & 0 & D & 8 & 1 & B & 3 & B & 5 & E & 1 & 9 & E & 0 & 3 & 6 & C & F & 3 & C & D & D & 1 & 4 & 6 & A & 1 & D & 7 & 7 \\
1 & 4 & 0 & 7 & 6 & B & B & 7 & 6 & 1 & 7 & 3 & 7 & 5 & 8 & 0 & D & D & A & 3 & 6 & 1 & 5 & D & 1 & 5 & B & B & 6 & 8 & A & 2 & 8 & 1 & 9 & E & 7 & 7 & 3 & 1 & 8 & 6 & 6 & 1 & B & 6 & 5 \\
0 & C & 0 & 1 & 8 & 3 & C & C & 6 & C & 1 & D & B & 5 & D & 6 & 3 & 5 & 6 & 1 & E & C & D & D & 0 & 1 & E & A & 0 & 3 & 6 & D & 8 & 6 & 3 & 5 & D & 8 & 6 & C & 3 & 3 & D & 5 & E & D & B
\end{array}\right].
\]
}
\item A binary self-orthogonal $[227, 9, 110]$ code which has an automorphism group of size $|\mathrm{Aut}|=2$ and a generator matrix
{\scriptsize
\setlength{\arraycolsep}{0.35pt}%
\renewcommand{\arraystretch}{0.9}%
\[
\resizebox{\linewidth}{!}{$\displaystyle G_{\mathrm{SO},19}=
\left[\begin{array}{@{}*{57}{c}@{}}
8 & 2 & 7 & 8 & C & 4 & 8 & 5 & D & C & 3 & 4 & 0 & F & 5 & 9 & 4 & 5 & 6 & 2 & 7 & 9 & 0 & 6 & 2 & D & F & F & B & 2 & 5 & F & D & 5 & 3 & B & 2 & 3 & 3 & 4 & 5 & C & 1 & 7 & D & 8 & 5 & 3 & 9 & F & 2 & 6 & B & A & 2 & A & 0 \\
4 & 1 & 7 & C & 4 & 2 & 5 & 2 & E & 6 & 1 & E & 0 & 5 & A & D & A & 2 & 3 & 1 & 3 & C & A & 3 & 1 & 6 & F & 7 & D & D & 2 & D & E & B & 9 & D & 1 & 1 & D & A & 0 & E & 1 & B & E & 4 & 2 & D & C & F & 9 & 2 & 5 & D & 1 & 5 & 0 \\
2 & 0 & F & E & 0 & 1 & 3 & 6 & 8 & B & 0 & 8 & F & F & 2 & 8 & D & 1 & E & 7 & 6 & 1 & A & E & 6 & 4 & 8 & 4 & 1 & 1 & 6 & A & F & 5 & C & E & 7 & 7 & 2 & D & 3 & 8 & F & 2 & 0 & 5 & E & E & E & 4 & 3 & 6 & D & 1 & F & 5 & 0 \\
1 & 2 & 4 & 7 & E & 4 & 1 & 1 & 6 & 1 & B & 7 & 8 & D & 3 & 2 & 2 & D & E & E & 3 & 6 & 0 & E & F & 0 & 4 & 2 & 4 & 5 & 1 & 6 & A & F & D & C & E & 7 & 2 & 2 & D & F & 8 & 1 & 2 & D & 5 & C & E & E & C & 2 & 2 & D & A & F & 0 \\
0 & 9 & 6 & 3 & D & 2 & 1 & 7 & 4 & 0 & D & 8 & 3 & 9 & 6 & 7 & 1 & 6 & 8 & 8 & E & 4 & D & 8 & 9 & 7 & D & E & D & 9 & 7 & 7 & 5 & 6 & E & E & 8 & C & 5 & 1 & 7 & 0 & 3 & F & 6 & 1 & 5 & 6 & 7 & 4 & 9 & E & E & 9 & A & 8 & 0 \\
0 & 4 & F & 1 & C & 9 & 0 & B & A & 8 & 6 & 8 & 1 & E & B & 2 & 8 & B & C & 4 & 7 & 2 & 4 & C & 5 & B & E & F & 6 & C & B & B & A & A & 7 & 7 & 4 & 6 & 6 & 8 & B & 8 & 0 & F & B & 0 & A & F & 3 & A & 4 & F & 7 & 4 & 5 & 4 & 0 \\
8 & 0 & 9 & 3 & 7 & 7 & 4 & 1 & F & 4 & 0 & F & 6 & E & 7 & 3 & C & 1 & 1 & 6 & 2 & 5 & F & 7 & B & 9 & F & 2 & 5 & 0 & C & 6 & 4 & 6 & B & E & D & 6 & 0 & A & C & 4 & 0 & A & 4 & 5 & A & B & E & 9 & 0 & D & 5 & 3 & B & 8 & 2 \\
4 & 0 & D & A & E & C & E & E & F & E & 0 & F & 2 & 6 & B & 4 & 2 & 1 & E & 2 & 8 & 8 & F & 3 & 9 & A & F & C & 8 & 3 & 5 & 9 & 6 & 4 & E & 1 & 4 & 2 & C & F & B & 9 & F & 0 & 9 & 0 & 8 & 2 & 1 & E & 7 & 5 & 0 & 5 & 9 & B & 6 \\
\end{array}\right].
$}
\]
}
\item A binary self-orthogonal $[255, 9, 124]$ code which has an automorphism group of size $|\mathrm{Aut}|=1981681804247040$ and a generator matrix
{\scriptsize
\setlength{\arraycolsep}{0.35pt}%
\renewcommand{\arraystretch}{0.9}%
\[
\resizebox{\linewidth}{!}{$\displaystyle G_{\mathrm{SO},20}=
\left[\begin{array}{@{}*{64}{c}@{}}
F & 0 & 0 & F & 0 & F & F & 0 & 0 & F & F & 0 & F & 0 & 0 & F & 0 & F & F & 0 & F & 0 & 0 & F & F & 0 & 0 & F & 0 & F & F & 0 & 0 & F & F & 0 & F & 0 & 0 & F & F & 0 & 0 & F & 0 & F & F & 0 & F & 0 & 0 & F & 0 & F & F & 0 & 0 & F & F & 0 & F & 0 & 0 & 0 \\
0 & F & 0 & F & 0 & F & 0 & F & 0 & F & 0 & F & 0 & F & 0 & F & 0 & F & 0 & F & 0 & F & 0 & F & 0 & F & 0 & F & 0 & F & 0 & F & 0 & F & 0 & F & 0 & F & 0 & F & 0 & F & 0 & F & 0 & F & 0 & F & 0 & F & 0 & F & 0 & F & 0 & F & 0 & F & 0 & F & 0 & F & 0 & 0 \\
0 & 0 & F & F & 0 & 0 & F & F & 0 & 0 & F & F & 0 & 0 & F & F & 0 & 0 & F & F & 0 & 0 & F & F & 0 & 0 & F & F & 0 & 0 & F & F & 0 & 0 & F & F & 0 & 0 & F & F & 0 & 0 & F & F & 0 & 0 & F & F & 0 & 0 & F & F & 0 & 0 & F & F & 0 & 0 & F & F & 0 & 0 & F & 0 \\
0 & 0 & 0 & 0 & F & F & F & F & 0 & 0 & 0 & 0 & F & F & F & F & 0 & 0 & 0 & 0 & F & F & F & F & 0 & 0 & 0 & 0 & F & F & F & F & 0 & 0 & 0 & 0 & F & F & F & F & 0 & 0 & 0 & 0 & F & F & F & F & 0 & 0 & 0 & 0 & F & F & F & F & 0 & 0 & 0 & 0 & F & F & F & 0 \\
0 & 0 & 0 & 0 & 0 & 0 & 0 & 0 & F & F & F & F & F & F & F & F & 0 & 0 & 0 & 0 & 0 & 0 & 0 & 0 & F & F & F & F & F & F & F & F & 0 & 0 & 0 & 0 & 0 & 0 & 0 & 0 & F & F & F & F & F & F & F & F & 0 & 0 & 0 & 0 & 0 & 0 & 0 & 0 & F & F & F & F & F & F & F & 0 \\
0 & 0 & 0 & 0 & 0 & 0 & 0 & 0 & 0 & 0 & 0 & 0 & 0 & 0 & 0 & 0 & F & F & F & F & F & F & F & F & F & F & F & F & F & F & F & F & 0 & 0 & 0 & 0 & 0 & 0 & 0 & 0 & 0 & 0 & 0 & 0 & 0 & 0 & 0 & 0 & F & F & F & F & F & F & F & F & F & F & F & F & F & F & F & 0 \\
0 & 0 & 0 & 0 & 0 & 0 & 0 & 0 & 0 & 0 & 0 & 0 & 0 & 0 & 0 & 0 & 0 & 0 & 0 & 0 & 0 & 0 & 0 & 0 & 0 & 0 & 0 & 0 & 0 & 0 & 0 & 0 & F & F & F & F & F & F & F & F & F & F & F & F & F & F & F & F & F & F & F & F & F & F & F & F & F & F & F & F & F & F & F & 0 \\
\end{array}\right].
$}
\]
}
\item A binary self-orthogonal $[70, 10, 30]$ code which has an automorphism group of size $|\mathrm{Aut}|=1512$ and a generator matrix
{\scriptsize
\setlength{\arraycolsep}{0.35pt}%
\renewcommand{\arraystretch}{0.9}%
\[
G_{\mathrm{SO},21}=
\left[\begin{array}{@{}*{18}{c}@{}}
8 & 0 & 4 & 6 & F & 6 & 4 & 1 & B & 5 & 4 & E & A & 5 & 6 & 5 & E & 0 \\
4 & 1 & 3 & 8 & A & 2 & 2 & 6 & 0 & F & 9 & D & C & 7 & 2 & 7 & E & 0 \\
2 & 1 & C & 1 & 7 & E & 5 & 4 & 6 & 7 & B & A & D & 3 & 6 & 2 & 0 & 0 \\
1 & 0 & A & 6 & 4 & 9 & 6 & B & 8 & 6 & 9 & 3 & C & C & D & 5 & E & 0 \\
0 & 9 & 4 & 8 & F & D & B & 3 & 1 & 6 & 7 & 3 & 7 & 3 & F & F & E & 0 \\
0 & 5 & F & 9 & 5 & 1 & 9 & E & E & B & 4 & D & 8 & 9 & 0 & E & 0 & 0 \\
0 & 2 & B & A & 5 & E & 8 & E & C & 0 & E & 8 & 6 & 1 & E & 3 & E & 0 \\
8 & 0 & 3 & D & A & D & C & 0 & 9 & 3 & 3 & A & 6 & 3 & 2 & 3 & 8 & 4 \\
4 & 0 & 2 & 3 & 7 & B & 2 & 0 & D & A & A & 7 & 5 & 2 & B & 2 & C & C \\
2 & 0 & 2 & C & 1 & 0 & 5 & 0 & F & E & 6 & 9 & C & A & 7 & B & F & C
\end{array}\right].
\]
}
\item A binary self-orthogonal $[86, 10, 38]$ code which has an automorphism group of size $|\mathrm{Aut}|=2$ and a generator matrix
{\scriptsize
\setlength{\arraycolsep}{0.35pt}%
\renewcommand{\arraystretch}{0.9}%
\[
G_{\mathrm{SO},22}=
\left[\begin{array}{@{}*{22}{c}@{}}
8 & 1 & E & 3 & 4 & 2 & A & 0 & F & 7 & C & F & E & B & 5 & 3 & 6 & 0 & 3 & 2 & 8 & 0 \\
4 & 0 & 3 & 5 & C & E & 3 & 5 & C & A & F & A & 0 & D & 2 & 5 & 6 & B & F & 2 & 0 & 0 \\
2 & 3 & 0 & 0 & F & 7 & 8 & E & 7 & C & C & B & 5 & 9 & A & E & 2 & 1 & 7 & C & 0 & 0 \\
1 & 6 & 8 & 0 & E & D & 1 & 8 & 8 & 7 & 7 & D & 6 & D & 9 & 6 & C & 3 & 2 & B & 8 & 0 \\
0 & F & 5 & 5 & 4 & E & 3 & 8 & 2 & B & 5 & 6 & 5 & A & 7 & E & 1 & 8 & 1 & F & 0 & 0 \\
0 & 0 & 0 & E & 8 & F & 4 & E & 4 & 6 & 3 & C & B & 7 & A & B & B & 5 & D & 1 & 8 & 0 \\
8 & 0 & 0 & 5 & 8 & 3 & 5 & 8 & F & C & 1 & 7 & 5 & E & B & 6 & E & 2 & 9 & 2 & 4 & 8 \\
4 & 0 & 0 & 6 & 3 & 8 & 9 & B & F & D & 8 & 6 & 3 & 3 & C & D & 8 & 1 & C & 9 & 7 & 4 \\
2 & 0 & 0 & 0 & C & 1 & 3 & B & 5 & B & 4 & E & 3 & C & 5 & 3 & 7 & 9 & 5 & 9 & D & 0 \\
1 & 4 & 0 & 3 & 7 & A & 3 & 3 & C & 0 & 3 & E & E & 6 & 0 & 3 & E & C & 8 & A & 8 & C
\end{array}\right].
\]
}
\item A binary self-orthogonal $[102, 10, 46]$ code which has an automorphism group of size $|\mathrm{Aut}|=2$ and a generator matrix
{\scriptsize
\setlength{\arraycolsep}{0.35pt}%
\renewcommand{\arraystretch}{0.9}%
\[
G_{\mathrm{SO},23}=
\left[\begin{array}{@{}*{26}{c}@{}}
8 & 0 & 1 & E & 5 & B & 6 & 6 & F & A & 3 & 5 & D & C & 9 & E & A & E & A & 4 & 3 & 7 & 8 & 0 & 0 & 0 \\
5 & 0 & D & E & 4 & 0 & A & E & A & 7 & C & 6 & D & 0 & 6 & D & 9 & 8 & F & 1 & D & D & 0 & D & 0 & 0 \\
3 & 0 & 2 & 8 & B & 0 & A & 5 & E & 6 & D & B & 5 & E & 8 & F & 2 & E & 5 & 3 & 0 & C & E & A & 4 & 0 \\
0 & A & E & E & B & 3 & 7 & E & 6 & B & E & 7 & 6 & D & 2 & A & 8 & 1 & 4 & 3 & 4 & 0 & 8 & 6 & 8 & 0 \\
0 & 6 & F & 0 & 0 & 5 & 7 & 3 & C & B & 3 & A & D & C & E & 3 & 4 & 4 & 8 & F & 3 & A & 6 & A & 8 & 0 \\
0 & 1 & 5 & D & 9 & E & 6 & D & 3 & 7 & B & 3 & 6 & B & 9 & 1 & F & 2 & 1 & 8 & 4 & 6 & 5 & 0 & C & 0 \\
8 & 0 & 0 & 3 & 2 & E & B & D & A & A & 1 & C & 3 & F & 5 & 5 & B & 6 & 2 & 9 & 9 & 0 & B & 6 & 7 & C \\
5 & 0 & 0 & 0 & 5 & 8 & 0 & E & 4 & 6 & 5 & A & 8 & 7 & B & D & B & 2 & F & E & B & 6 & 3 & 0 & C & 8 \\
3 & 0 & 0 & 3 & B & 3 & 6 & 6 & 7 & B & 0 & 9 & B & 3 & 5 & B & 7 & 1 & 7 & 5 & 8 & A & D & 0 & 2 & 4 \\
0 & A & 0 & 3 & 5 & D & B & 0 & C & B & 3 & 0 & 6 & B & 9 & E & C & 1 & 1 & E & E & C & B & 0 & 9 & C
\end{array}\right].
\]
}
\item A binary self-orthogonal $[147, 10, 68]$ code which has an automorphism group of size $|\mathrm{Aut}|=21$ and a generator matrix
{\scriptsize
\setlength{\arraycolsep}{0.35pt}%
\renewcommand{\arraystretch}{0.9}%
\[
G_{\mathrm{SO},24}=
\left[\begin{array}{@{}*{37}{c}@{}}
9 & 0 & 9 & 2 & 9 & 8 & 4 & B & 9 & 0 & A & 1 & 3 & C & D & E & 5 & B & 3 & B & E & 4 & 9 & 3 & 2 & 3 & 3 & 7 & 1 & 3 & 6 & 1 & 2 & A & C & F & 0 \\
5 & 8 & 7 & D & F & 7 & 2 & B & C & 5 & 5 & 2 & A & A & D & A & 9 & 4 & E & 9 & 8 & 2 & F & 4 & 0 & F & 7 & A & 6 & 4 & 5 & 5 & 3 & 7 & 5 & C & 0 \\
2 & 8 & C & 1 & A & 4 & A & 1 & F & E & 2 & 6 & 9 & 8 & A & 4 & B & 6 & 1 & A & 2 & 2 & 6 & 9 & D & C & 7 & 3 & B & 5 & 2 & A & 9 & 3 & 4 & F & 0 \\
0 & 4 & 0 & F & C & 0 & F & 8 & 0 & F & 0 & F & E & 1 & F & 0 & 7 & F & F & F & 0 & F & 0 & 1 & E & 0 & F & 8 & 0 & 6 & 0 & 3 & F & 8 & 1 & 3 & 0 \\
0 & 2 & 0 & 0 & 3 & F & F & 8 & 0 & 0 & F & F & E & 0 & 0 & F & F & F & F & F & F & 0 & 0 & 1 & F & F & 0 & 0 & 0 & 7 & F & C & 0 & 0 & 1 & C & 0 \\
0 & 1 & F & F & F & F & F & 8 & 0 & 0 & 0 & 0 & 1 & F & F & F & F & F & F & F & F & F & F & E & 0 & 0 & 0 & 0 & 0 & 7 & F & F & F & F & E & 0 & 0 \\
0 & 0 & 0 & 0 & 0 & 0 & 0 & 7 & F & F & F & F & F & F & F & F & F & F & F & F & F & F & F & F & F & F & F & F & F & 8 & 0 & 0 & 0 & 0 & 0 & 0 & 0 \\
8 & 0 & B & 4 & C & 1 & F & 3 & 4 & 1 & D & 5 & D & 3 & 6 & A & D & 6 & 9 & 8 & 9 & 5 & B & D & C & A & E & E & 2 & F & A & C & B & 4 & B & 6 & A \\
4 & 0 & 4 & A & 6 & 9 & E & B & 2 & 7 & 1 & 5 & B & B & 1 & D & 9 & A & A & A & 1 & D & 3 & B & 7 & E & 6 & 4 & 5 & D & 5 & B & 2 & F & 2 & F & 2 \\
2 & 0 & D & 0 & 6 & 3 & D & 9 & C & D & 1 & 5 & 6 & 6 & D & 7 & 3 & 5 & 8 & 5 & 3 & 3 & 7 & 7 & B & B & 4 & B & 4 & F & 1 & 6 & E & A & B & A & 6
\end{array}\right].
\]
}
\item A binary self-orthogonal $[196, 10, 94]$ code which has an automorphism group of size $|\mathrm{Aut}|=4608$ and a generator matrix
{\scriptsize
\setlength{\arraycolsep}{0.35pt}%
\renewcommand{\arraystretch}{0.9}%
\[
\resizebox{\linewidth}{!}{$\displaystyle G_{\mathrm{SO},25}=
\left[\begin{array}{@{}*{49}{c}@{}}
A & 0 & 5 & 0 & 3 & 6 & A & E & E & 1 & B & 6 & A & 1 & B & B & 6 & D & C & 2 & 8 & 6 & 1 & 8 & 1 & A & B & C & 2 & B & B & 7 & 3 & 7 & 7 & 3 & 0 & D & E & 6 & 0 & 5 & 7 & 9 & B & C & 3 & 6 & 0 \\
6 & 0 & 5 & 0 & D & E & 1 & 4 & 0 & A & E & E & 6 & F & 3 & 7 & 6 & E & 1 & A & B & B & 6 & E & 6 & 0 & 6 & D & B & 5 & 7 & 9 & 8 & 6 & 1 & B & B & 8 & 5 & 7 & 8 & 6 & D & 4 & 0 & D & 4 & 5 & 0 \\
1 & 4 & 0 & A & E & E & 6 & F & 3 & 7 & 6 & E & 1 & A & B & B & 6 & E & 6 & 0 & 6 & D & B & 5 & 7 & 9 & 8 & 6 & 1 & B & B & 8 & 5 & 7 & 8 & 6 & D & 4 & 0 & D & 4 & 3 & 0 & 2 & 8 & 6 & F & 5 & 0 \\
0 & C & 0 & A & 1 & B & C & 2 & 8 & 1 & 5 & D & C & D & E & 6 & E & D & C & 3 & 5 & 7 & 6 & D & C & C & 0 & D & B & 6 & A & F & 3 & 0 & C & 3 & 7 & 7 & 0 & A & F & 0 & D & A & 8 & 1 & A & D & 0 \\
0 & 2 & 8 & 1 & B & 5 & 7 & 7 & 0 & D & B & 5 & 0 & D & D & B & 6 & E & 1 & 4 & 3 & 0 & C & 0 & D & 5 & E & 1 & 5 & D & B & 9 & B & B & 9 & 8 & 6 & F & 3 & 0 & 2 & B & C & D & E & 1 & A & E & 0 \\
0 & 1 & D & 1 & 7 & 5 & D & 6 & B & 1 & 9 & D & 1 & 8 & 0 & 7 & B & 0 & 7 & A & E & C & F & 5 & A & 3 & 3 & D & 9 & D & 6 & 2 & D & 1 & 6 & B & 6 & 3 & 0 & 7 & 5 & B & 6 & 2 & E & C & 0 & 3 & 0 \\
0 & 0 & 3 & A & 2 & E & B & A & D & 6 & 3 & 3 & A & 3 & 0 & 0 & F & 6 & 0 & F & 5 & D & 9 & E & B & 4 & 6 & 7 & B & 3 & A & C & 5 & A & 2 & D & 6 & C & 6 & 0 & E & B & 6 & C & 5 & D & 8 & 3 & 0 \\
0 & 0 & 0 & 7 & 4 & 5 & D & 7 & 5 & A & C & 6 & 7 & 4 & 6 & 0 & 1 & E & C & 1 & E & B & B & 3 & D & 6 & 8 & C & F & 6 & 7 & 5 & 8 & B & 4 & 5 & A & D & 8 & C & 1 & D & 6 & D & 8 & B & B & 3 & 0 \\
\end{array}\right].
$}
\]
}
\item A binary self-orthogonal $[188, 11, 88]$ code which has an automorphism group of size $|\mathrm{Aut}|=2$ and a generator matrix
{\scriptsize
\setlength{\arraycolsep}{0.35pt}%
\renewcommand{\arraystretch}{0.9}%
\[
G_{\mathrm{SO},1}=
\left[\begin{array}{@{}*{47}{c}@{}}
B & 9 & 8 & 0 & E & E & D & 4 & 0 & 7 & 8 & 3 & 0 & D & 8 & A & 2 & B & B & B & 0 & C & D & E & 7 & B & 5 & C & 2 & B & 6 & F & 6 & A & F & 6 & 6 & C & 3 & 0 & F & 6 & C & 2 & 8 & 0 & 0 \\
A & 0 & 0 & 0 & 2 & B & 0 & 0 & 3 & C & 3 & 3 & 6 & C & 6 & D & 6 & E & B & 4 & 0 & B & 9 & D & 6 & E & E & A & C & 6 & 1 & 5 & B & D & 7 & 5 & 7 & 4 & 3 & A & C & 3 & 6 & 0 & 7 & 4 & 9 \\
6 & 0 & 0 & 0 & C & 3 & B & A & D & 7 & 6 & B & A & 2 & E & 1 & 6 & D & 6 & C & 3 & 6 & E & B & 1 & 4 & 3 & B & 5 & 8 & D & B & 0 & C & 1 & D & C & 1 & 8 & B & 4 & 0 & C & D & D & 8 & 8 \\
1 & 4 & 0 & 0 & 2 & E & 6 & 0 & 3 & B & B & 5 & 0 & 1 & E & 0 & C & 3 & 6 & 2 & 8 & A & E & E & C & 3 & 3 & 7 & 9 & E & D & 7 & 0 & A & D & B & D & A & B & D & 9 & B & 0 & C & 4 & 4 & A \\
0 & C & 0 & 0 & D & B & C & D & 8 & D & 8 & 6 & D & 6 & B & D & 4 & 0 & C & 1 & B & 0 & 3 & 6 & 7 & 6 & B & C & 3 & 3 & C & 0 & 6 & D & 9 & E & 7 & 9 & B & A & 2 & 8 & D & 4 & 7 & D & 9 \\
0 & 2 & 8 & 0 & E & D & 7 & 6 & E & C & 2 & E & 6 & E & B & 0 & 1 & B & B & 4 & 6 & C & 2 & B & A & 2 & B & 7 & 6 & D & 1 & 4 & 5 & 0 & 3 & 3 & C & E & E & 6 & 3 & 0 & C & C & 4 & 5 & 4 \\
0 & 0 & 5 & 0 & 1 & D & A & E & D & D & 8 & 5 & C & D & D & 6 & 0 & 3 & 7 & 6 & 8 & D & 8 & 5 & 7 & 4 & 5 & 6 & E & D & A & 2 & 8 & A & 0 & 6 & 7 & 9 & D & C & C & 6 & 1 & 9 & B & 8 & 8 \\
0 & 0 & 3 & 0 & E & E & B & 5 & D & A & E & 8 & B & 8 & 5 & B & 5 & B & 0 & D & B & A & C & 5 & 0 & E & D & 6 & 3 & 6 & C & 3 & 0 & 7 & 7 & 0 & 6 & 2 & D & 0 & 3 & 3 & 6 & 2 & C & A & 4 \\
\end{array}\right].
\]
}
\item A binary self-orthogonal $[188, 11, 88]$ code which has an automorphism group of size $|\mathrm{Aut}|=1$ and a generator matrix
{\scriptsize
\setlength{\arraycolsep}{0.35pt}%
\renewcommand{\arraystretch}{0.9}%
\[
G_{\mathrm{SO},2}=
\left[\begin{array}{@{}*{47}{c}@{}}
B & 9 & 8 & 0 & E & E & D & 4 & 0 & 7 & 8 & 3 & 0 & D & 8 & A & 2 & B & B & B & 0 & C & D & E & 7 & B & 5 & C & 2 & B & 6 & F & 6 & A & F & 6 & 6 & C & 3 & 0 & F & 6 & C & 2 & 8 & 0 & 0 \\
A & 0 & 0 & 0 & 2 & B & 0 & 0 & 3 & C & 3 & 3 & 6 & C & 6 & D & 6 & E & B & 4 & 0 & B & 9 & D & 6 & E & E & A & C & 6 & 1 & 5 & B & D & 7 & 5 & 7 & 4 & 3 & A & C & 3 & 6 & 0 & 1 & 4 & F \\
6 & 0 & 0 & 0 & C & 3 & B & A & D & 7 & 6 & B & A & 2 & E & 1 & 6 & D & 6 & C & 3 & 6 & E & B & 1 & 4 & 3 & B & 5 & 8 & D & B & 0 & C & 1 & D & C & 1 & 8 & B & 4 & 0 & C & D & A & 4 & A \\
1 & 4 & 0 & 0 & 2 & E & 6 & 0 & 3 & B & B & 5 & 0 & 1 & E & 0 & C & 3 & 6 & 2 & 8 & A & E & E & C & 3 & 3 & 7 & 9 & E & D & 7 & 0 & A & D & B & D & A & B & D & 9 & B & 0 & C & 0 & C & 6 \\
0 & C & 0 & 0 & D & B & C & D & 8 & D & 8 & 6 & D & 6 & B & D & 4 & 0 & C & 1 & B & 0 & 3 & 6 & 7 & 6 & B & C & 3 & 3 & C & 0 & 6 & D & 9 & E & 7 & 9 & B & A & 2 & 8 & D & 4 & 0 & 3 & 5 \\
0 & 2 & 8 & 0 & E & D & 7 & 6 & E & C & 2 & E & 6 & E & B & 0 & 1 & B & B & 4 & 6 & C & 2 & B & A & 2 & B & 7 & 6 & D & 1 & 4 & 5 & 0 & 3 & 3 & C & E & E & 6 & 3 & 0 & C & C & 0 & 3 & 6 \\
0 & 0 & 5 & 0 & 1 & D & A & E & D & D & 8 & 5 & C & D & D & 6 & 0 & 3 & 7 & 6 & 8 & D & 8 & 5 & 7 & 4 & 5 & 6 & E & D & A & 2 & 8 & A & 0 & 6 & 7 & 9 & D & C & C & 6 & 1 & 9 & B & 4 & 8 \\
0 & 0 & 3 & 0 & E & E & B & 5 & D & A & E & 8 & B & 8 & 5 & B & 5 & B & 0 & D & B & A & C & 5 & 0 & E & D & 6 & 3 & 6 & C & 3 & 0 & 7 & 7 & 0 & 6 & 2 & D & 0 & 3 & 3 & 6 & 2 & E & B & 7 \\
\end{array}\right].
\]
}
\item A binary self-orthogonal $[188, 11, 88]$ code which has an automorphism group of size $|\mathrm{Aut}|=1$ and a generator matrix
{\scriptsize
\setlength{\arraycolsep}{0.35pt}%
\renewcommand{\arraystretch}{0.9}%
\[
G_{\mathrm{SO},3}=
\left[\begin{array}{@{}*{47}{c}@{}}
B & 9 & 8 & 0 & E & E & D & 4 & 0 & 7 & 8 & 3 & 0 & D & 8 & A & 2 & B & B & B & 0 & C & D & E & 7 & B & 5 & C & 2 & B & 6 & F & 6 & A & F & 6 & 6 & C & 3 & 0 & F & 6 & C & 2 & 8 & 0 & 0 \\
7 & A & D & 6 & D & B & 1 & 9 & E & B & A & E & C & 1 & D & 1 & 9 & D & 6 & C & 5 & A & C & 5 & D & 6 & 8 & B & B & 3 & A & 2 & D & 1 & 8 & 3 & D & 9 & E & C & 0 & 3 & A & C & 6 & 0 & 0 \\
A & 0 & 0 & 0 & 2 & B & 0 & 0 & 3 & C & 3 & 3 & 6 & C & 6 & D & 6 & E & B & 4 & 0 & B & 9 & D & 6 & E & E & A & C & 6 & 1 & 5 & B & D & 7 & 5 & 7 & 4 & 3 & A & C & 3 & 6 & 0 & 3 & 9 & 9 \\
6 & 0 & 0 & 0 & C & 3 & B & A & D & 7 & 6 & B & A & 2 & E & 1 & 6 & D & 6 & C & 3 & 6 & E & B & 1 & 4 & 3 & B & 5 & 8 & D & B & 0 & C & 1 & D & C & 1 & 8 & B & 4 & 0 & C & D & 9 & 8 & 5 \\
1 & 4 & 0 & 0 & 2 & E & 6 & 0 & 3 & B & B & 5 & 0 & 1 & E & 0 & C & 3 & 6 & 2 & 8 & A & E & E & C & 3 & 3 & 7 & 9 & E & D & 7 & 0 & A & D & B & D & A & B & D & 9 & B & 0 & C & 3 & A & 0 \\
0 & C & 0 & 0 & D & B & C & D & 8 & D & 8 & 6 & D & 6 & B & D & 4 & 0 & C & 1 & B & 0 & 3 & 6 & 7 & 6 & B & C & 3 & 3 & C & 0 & 6 & D & 9 & E & 7 & 9 & B & A & 2 & 8 & D & 4 & 3 & 9 & F \\
0 & 2 & 8 & 0 & E & D & 7 & 6 & E & C & 2 & E & 6 & E & B & 0 & 1 & B & B & 4 & 6 & C & 2 & B & A & 2 & B & 7 & 6 & D & 1 & 4 & 5 & 0 & 3 & 3 & C & E & E & 6 & 3 & 0 & C & C & 2 & C & 2 \\
0 & 0 & 5 & 0 & 1 & D & A & E & D & D & 8 & 5 & C & D & D & 6 & 0 & 3 & 7 & 6 & 8 & D & 8 & 5 & 7 & 4 & 5 & 6 & E & D & A & 2 & 8 & A & 0 & 6 & 7 & 9 & D & C & C & 6 & 1 & 9 & 9 & 0 & D \\
\end{array}\right].
\]
}
\item A binary self-orthogonal $[189, 11, 88]$ code which has an automorphism group of size $|\mathrm{Aut}|=1$ and a generator matrix
{\scriptsize
\setlength{\arraycolsep}{0.35pt}%
\renewcommand{\arraystretch}{0.9}%
\[
G_{\mathrm{SO},4}=
\left[\begin{array}{@{}*{48}{c}@{}}
A & E & B & 0 & F & 3 & 0 & E & 8 & 7 & 7 & 3 & 6 & C & 3 & B & 6 & E & C & C & 6 & D & 4 & 5 & B & 4 & 3 & 7 & A & E & 0 & 2 & 8 & 7 & A & 8 & A & 1 & B & 6 & E & 8 & 1 & A & D & 0 & 0 & 0 \\
6 & D & E & 6 & C & 6 & C & 3 & 6 & B & 5 & E & A & 0 & 6 & 0 & D & 8 & 1 & B & 3 & B & 5 & E & 1 & 9 & E & 0 & 3 & 6 & C & F & 3 & C & D & D & 1 & 4 & 6 & A & 1 & D & 7 & 4 & 3 & 0 & 0 & 0 \\
1 & 7 & 3 & 0 & 1 & D & D & A & 8 & 0 & F & 0 & 6 & 1 & B & 1 & 4 & 5 & 7 & 7 & 6 & 1 & 9 & B & C & F & 6 & B & 8 & 5 & 6 & D & E & D & 5 & E & C & D & 8 & 6 & 1 & E & D & 8 & 5 & 0 & 0 & 0 \\
A & 0 & 0 & 0 & 2 & B & 0 & 0 & 3 & C & 3 & 3 & 6 & C & 6 & D & 6 & E & B & 4 & 0 & B & 9 & D & 6 & E & E & A & C & 6 & 1 & 5 & B & D & 7 & 5 & 7 & 4 & 3 & A & C & 3 & 6 & 0 & 3 & C & B & 8 \\
6 & 0 & 0 & 0 & C & 3 & B & A & D & 7 & 6 & B & A & 2 & E & 1 & 6 & D & 6 & C & 3 & 6 & E & B & 1 & 4 & 3 & B & 5 & 8 & D & B & 0 & C & 1 & D & C & 1 & 8 & B & 4 & 0 & C & D & 8 & 8 & E & 0 \\
1 & 4 & 0 & 0 & 2 & E & 6 & 0 & 3 & B & B & 5 & 0 & 1 & E & 0 & C & 3 & 6 & 2 & 8 & A & E & E & C & 3 & 3 & 7 & 9 & E & D & 7 & 0 & A & D & B & D & A & B & D & 9 & B & 0 & C & 3 & 6 & 1 & 8 \\
0 & C & 0 & 0 & D & B & C & D & 8 & D & 8 & 6 & D & 6 & B & D & 4 & 0 & C & 1 & B & 0 & 3 & 6 & 7 & 6 & B & C & 3 & 3 & C & 0 & 6 & D & 9 & E & 7 & 9 & B & A & 2 & 8 & D & 4 & 3 & 7 & B & 0 \\
0 & 2 & 8 & 0 & E & D & 7 & 6 & E & C & 2 & E & 6 & E & B & 0 & 1 & B & B & 4 & 6 & C & 2 & B & A & 2 & B & 7 & 6 & D & 1 & 4 & 5 & 0 & 3 & 3 & C & E & E & 6 & 3 & 0 & C & C & 3 & 1 & 4 & 0 \\
\end{array}\right].
\]
}
\item A binary self-orthogonal $[189, 11, 88]$ code which has an automorphism group of size $|\mathrm{Aut}|=2$ and a generator matrix
{\scriptsize
\setlength{\arraycolsep}{0.35pt}%
\renewcommand{\arraystretch}{0.9}%
\[
G_{\mathrm{SO},5}=
\left[\begin{array}{@{}*{48}{c}@{}}
A & E & B & 0 & F & 3 & 0 & E & 8 & 7 & 7 & 3 & 6 & C & 3 & B & 6 & E & C & C & 6 & D & 4 & 5 & B & 4 & 3 & 7 & A & E & 0 & 2 & 8 & 7 & A & 8 & A & 1 & B & 6 & E & 8 & 1 & A & D & 0 & 0 & 0 \\
6 & D & E & 6 & C & 6 & C & 3 & 6 & B & 5 & E & A & 0 & 6 & 0 & D & 8 & 1 & B & 3 & B & 5 & E & 1 & 9 & E & 0 & 3 & 6 & C & F & 3 & C & D & D & 1 & 4 & 6 & A & 1 & D & 7 & 4 & 3 & 0 & 0 & 0 \\
1 & 7 & 3 & 0 & 1 & D & D & A & 8 & 0 & F & 0 & 6 & 1 & B & 1 & 4 & 5 & 7 & 7 & 6 & 1 & 9 & B & C & F & 6 & B & 8 & 5 & 6 & D & E & D & 5 & E & C & D & 8 & 6 & 1 & E & D & 8 & 5 & 0 & 0 & 0 \\
A & 0 & 0 & 0 & 2 & B & 0 & 0 & 3 & C & 3 & 3 & 6 & C & 6 & D & 6 & E & B & 4 & 0 & B & 9 & D & 6 & E & E & A & C & 6 & 1 & 5 & B & D & 7 & 5 & 7 & 4 & 3 & A & C & 3 & 6 & 0 & 3 & 8 & C & 8 \\
6 & 0 & 0 & 0 & C & 3 & B & A & D & 7 & 6 & B & A & 2 & E & 1 & 6 & D & 6 & C & 3 & 6 & E & B & 1 & 4 & 3 & B & 5 & 8 & D & B & 0 & C & 1 & D & C & 1 & 8 & B & 4 & 0 & C & D & 8 & E & A & 8 \\
1 & 4 & 0 & 0 & 2 & E & 6 & 0 & 3 & B & B & 5 & 0 & 1 & E & 0 & C & 3 & 6 & 2 & 8 & A & E & E & C & 3 & 3 & 7 & 9 & E & D & 7 & 0 & A & D & B & D & A & B & D & 9 & B & 0 & C & 3 & 8 & 2 & 0 \\
0 & C & 0 & 0 & D & B & C & D & 8 & D & 8 & 6 & D & 6 & B & D & 4 & 0 & C & 1 & B & 0 & 3 & 6 & 7 & 6 & B & C & 3 & 3 & C & 0 & 6 & D & 9 & E & 7 & 9 & B & A & 2 & 8 & D & 4 & 3 & D & C & 8 \\
0 & 2 & 8 & 0 & E & D & 7 & 6 & E & C & 2 & E & 6 & E & B & 0 & 1 & B & B & 4 & 6 & C & 2 & B & A & 2 & B & 7 & 6 & D & 1 & 4 & 5 & 0 & 3 & 3 & C & E & E & 6 & 3 & 0 & C & C & 3 & 6 & C & 0 \\
\end{array}\right].
\]
}
\item A binary self-orthogonal $[189, 11, 88]$ code which has an automorphism group of size $|\mathrm{Aut}|=1$ and a generator matrix
{\scriptsize
\setlength{\arraycolsep}{0.35pt}%
\renewcommand{\arraystretch}{0.9}%
\[
G_{\mathrm{SO},6}=
\left[\begin{array}{@{}*{48}{c}@{}}
A & 1 & E & A & 2 & 8 & 6 & D & B & A & 0 & 6 & B & 4 & 0 & 1 & 5 & D & 6 & 1 & E & 6 & 1 & D & 0 & E & E & 6 & D & 8 & 7 & 6 & D & D & 9 & 8 & D & A & 8 & B & 6 & 8 & 6 & F & 5 & C & 0 & 0 \\
6 & 2 & B & C & 1 & D & A & 0 & 5 & 6 & 2 & B & 7 & 8 & 5 & A & E & B & B & 6 & B & 0 & 0 & 6 & A & 3 & 3 & 1 & 4 & 0 & B & B & 6 & 6 & E & D & 6 & F & 5 & 7 & 9 & D & 0 & 1 & B & C & 0 & 0 \\
1 & 7 & 3 & 0 & 1 & D & D & A & 8 & 0 & F & 0 & 6 & 1 & B & 1 & 4 & 5 & 7 & 7 & 6 & 1 & 9 & B & C & F & 6 & B & 8 & 5 & 6 & D & E & D & 5 & E & C & D & 8 & 6 & 1 & E & D & 8 & 5 & 0 & 0 & 0 \\
0 & F & 5 & A & D & B & 6 & 3 & 3 & D & 7 & 5 & D & 8 & 3 & A & 3 & 3 & A & D & 8 & B & 5 & 8 & B & A & D & 1 & 7 & 6 & 7 & 4 & 5 & A & 3 & 0 & 7 & B & 3 & D & 8 & 0 & 7 & 5 & 8 & C & 0 & 0 \\
A & 0 & 0 & 0 & 2 & B & 0 & 0 & 3 & C & 3 & 3 & 6 & C & 6 & D & 6 & E & B & 4 & 0 & B & 9 & D & 6 & E & E & A & C & 6 & 1 & 5 & B & D & 7 & 5 & 7 & 4 & 3 & A & C & 3 & 6 & 0 & 3 & A & 8 & 8 \\
6 & 0 & 0 & 0 & C & 3 & B & A & D & 7 & 6 & B & A & 2 & E & 1 & 6 & D & 6 & C & 3 & 6 & E & B & 1 & 4 & 3 & B & 5 & 8 & D & B & 0 & C & 1 & D & C & 1 & 8 & B & 4 & 0 & C & D & 8 & B & 5 & 8 \\
1 & 4 & 0 & 0 & 2 & E & 6 & 0 & 3 & B & B & 5 & 0 & 1 & E & 0 & C & 3 & 6 & 2 & 8 & A & E & E & C & 3 & 3 & 7 & 9 & E & D & 7 & 0 & A & D & B & D & A & B & D & 9 & B & 0 & C & 3 & F & 7 & 8 \\
0 & C & 0 & 0 & D & B & C & D & 8 & D & 8 & 6 & D & 6 & B & D & 4 & 0 & C & 1 & B & 0 & 3 & 6 & 7 & 6 & B & C & 3 & 3 & C & 0 & 6 & D & 9 & E & 7 & 9 & B & A & 2 & 8 & D & 4 & 3 & B & A & 8 \\
\end{array}\right].
\]
}
\item A binary self-orthogonal $[191, 11, 90]$ code which has an automorphism group of size $|\mathrm{Aut}|=2$ and a generator matrix
{\scriptsize
\setlength{\arraycolsep}{0.35pt}%
\renewcommand{\arraystretch}{0.9}%
\[
\resizebox{\linewidth}{!}{$\displaystyle G_{\mathrm{SO},7}=
\left[\begin{array}{@{}*{48}{c}@{}}
A & 0 & 0 & 1 & 7 & 3 & 0 & 1 & D & D & A & 8 & 0 & F & 0 & 6 & 1 & B & 1 & 4 & 5 & 7 & 7 & 6 & 1 & 9 & B & C & F & 6 & B & 8 & 5 & 6 & D & E & D & 5 & E & C & D & 8 & 6 & 1 & E & D & 8 & 0 \\
6 & 0 & 0 & 6 & D & E & 6 & C & 6 & C & 3 & 6 & B & 5 & E & A & 0 & 6 & 0 & D & 8 & 1 & B & 3 & B & 5 & E & 1 & 9 & E & 0 & 3 & 6 & C & F & 3 & C & D & D & 1 & 4 & 6 & A & 1 & D & 7 & 4 & 0 \\
1 & 4 & 3 & D & 4 & 5 & 0 & D & B & 7 & 4 & 0 & D & 6 & 8 & 0 & 2 & B & A & C & 3 & C & C & 3 & A & 1 & D & C & D & B & 0 & E & D & B & B & 3 & 1 & B & 5 & 1 & 6 & D & 0 & D & E & B & 8 & 0 \\
0 & C & 0 & B & 4 & 3 & C & 3 & 6 & 0 & C & 6 & A & E & 8 & D & 9 & 8 & 6 & 3 & 0 & B & 8 & 6 & B & B & 5 & B & 8 & 6 & 0 & 2 & D & B & 6 & 8 & D & 6 & D & C & E & B & D & B & 5 & 6 & C & 0 \\
0 & 2 & B & A & C & 3 & C & C & 3 & A & 1 & D & C & D & B & 0 & E & D & B & B & 3 & 1 & B & 5 & 1 & 6 & D & 0 & D & E & B & 8 & 0 & A & 1 & E & A & 2 & 8 & 6 & D & B & A & 0 & 6 & B & 4 & 0 \\
0 & 1 & B & 7 & 9 & B & 1 & B & 0 & D & A & D & 7 & A & 8 & 1 & 8 & 3 & 6 & 0 & 6 & C & E & D & 7 & 8 & 6 & 7 & 8 & 0 & D & B & 3 & C & F & 3 & 7 & 4 & 5 & 1 & A & 8 & 7 & 5 & D & 0 & C & 0 \\
0 & 0 & 5 & C & C & 0 & 7 & 7 & 6 & A & 0 & 3 & C & 1 & 8 & 6 & C & 5 & 1 & 5 & D & D & 8 & 6 & 6 & F & 3 & D & A & E & 1 & 5 & B & 7 & B & 5 & 7 & B & 3 & 6 & 1 & 8 & 7 & B & 6 & 1 & 4 & 0 \\
A & 0 & 0 & 0 & 2 & B & 0 & 0 & 3 & C & 3 & 3 & 6 & C & 6 & D & 6 & E & B & 4 & 0 & B & 9 & D & 6 & E & E & A & C & 6 & 1 & 5 & B & D & 7 & 5 & 7 & 4 & 3 & A & C & 3 & 6 & 0 & 3 & A & E & 2 \\
\end{array}\right].
$}
\]
}
\item A binary self-orthogonal $[198, 11, 94]$ code which has an automorphism group of size $|\mathrm{Aut}|=9216$ and a generator matrix
{\scriptsize
\setlength{\arraycolsep}{0.35pt}%
\renewcommand{\arraystretch}{0.9}%
\[
\resizebox{\linewidth}{!}{$\displaystyle G_{\mathrm{SO},26}=
\left[\begin{array}{@{}*{50}{c}@{}}
A & 0 & 5 & 0 & 3 & 6 & A & E & E & 1 & B & 6 & A & 1 & B & B & 6 & D & C & 2 & 8 & 6 & 1 & 8 & 1 & A & B & C & 2 & B & B & 7 & 3 & 7 & 7 & 3 & 0 & D & E & 6 & 0 & 5 & 7 & 9 & B & C & 3 & 6 & 0 & 0 \\
6 & 0 & 5 & 0 & D & E & 1 & 4 & 0 & A & E & E & 6 & F & 3 & 7 & 6 & E & 1 & A & B & B & 6 & E & 6 & 0 & 6 & D & B & 5 & 7 & 9 & 8 & 6 & 1 & B & B & 8 & 5 & 7 & 8 & 6 & D & 4 & 0 & D & 4 & 5 & 0 & 0 \\
1 & 4 & 0 & A & E & E & 6 & F & 3 & 7 & 6 & E & 1 & A & B & B & 6 & E & 6 & 0 & 6 & D & B & 5 & 7 & 9 & 8 & 6 & 1 & B & B & 8 & 5 & 7 & 8 & 6 & D & 4 & 0 & D & 4 & 3 & 0 & 2 & 8 & 6 & F & 5 & 0 & 0 \\
0 & C & 0 & A & 1 & B & C & 2 & 8 & 1 & 5 & D & C & D & E & 6 & E & D & C & 3 & 5 & 7 & 6 & D & C & C & 0 & D & B & 6 & A & F & 3 & 0 & C & 3 & 7 & 7 & 0 & A & F & 0 & D & A & 8 & 1 & A & D & 0 & 0 \\
0 & 2 & 8 & 1 & B & 5 & 7 & 7 & 0 & D & B & 5 & 0 & D & D & B & 6 & E & 1 & 4 & 3 & 0 & C & 0 & D & 5 & E & 1 & 5 & D & B & 9 & B & B & 9 & 8 & 6 & F & 3 & 0 & 2 & B & C & D & E & 1 & A & E & 0 & 0 \\
0 & 1 & D & 1 & 7 & 5 & D & 6 & B & 1 & 9 & D & 1 & 8 & 0 & 7 & B & 0 & 7 & A & E & C & F & 5 & A & 3 & 3 & D & 9 & D & 6 & 2 & D & 1 & 6 & B & 6 & 3 & 0 & 7 & 5 & B & 6 & 2 & E & C & 0 & 3 & 0 & 0 \\
0 & 0 & 3 & A & 2 & E & B & A & D & 6 & 3 & 3 & A & 3 & 0 & 0 & F & 6 & 0 & F & 5 & D & 9 & E & B & 4 & 6 & 7 & B & 3 & A & C & 5 & A & 2 & D & 6 & C & 6 & 0 & E & B & 6 & C & 5 & D & 8 & 3 & 0 & 0 \\
0 & 0 & 0 & 7 & 4 & 5 & D & 7 & 5 & A & C & 6 & 7 & 4 & 6 & 0 & 1 & E & C & 1 & E & B & B & 3 & D & 6 & 8 & C & F & 6 & 7 & 5 & 8 & B & 4 & 5 & A & D & 8 & C & 1 & D & 6 & D & 8 & B & B & 3 & 0 & 0 \\
\end{array}\right].
$}
\]
}
\item A binary self-orthogonal $[230, 11, 110]$ code which has an automorphism group of size $|\mathrm{Aut}|=2$ and a generator matrix
{\scriptsize
\setlength{\arraycolsep}{0.35pt}%
\renewcommand{\arraystretch}{0.9}%
\[
\resizebox{\linewidth}{!}{$\displaystyle G_{\mathrm{SO},27}=
\left[\begin{array}{@{}*{58}{c}@{}}
A & 0 & 0 & 0 & 6 & 9 & 1 & 9 & D & 1 & 5 & E & 6 & D & 9 & A & B & 1 & 9 & E & A & 3 & 3 & 0 & E & D & A & D & 5 & D & 1 & 8 & 3 & 7 & 5 & E & B & 4 & C & 6 & A & 2 & B & 6 & C & 3 & 6 & D & E & 3 & 7 & A & 4 & 6 & A & C & 0 & 0 \\
6 & 0 & 1 & 9 & E & 6 & 7 & 6 & B & 6 & C & 0 & 0 & C & F & 0 & 0 & 7 & B & 0 & D & A & B & 0 & E & E & 9 & 1 & 5 & D & 0 & F & 5 & C & 1 & B & C & D & 5 & E & 0 & C & 6 & D & B & 6 & B & 9 & E & 0 & 0 & C & E & D & 7 & A & 8 & 0 \\
1 & 5 & 4 & 0 & 0 & F & 1 & A & B & C & C & 7 & 8 & 7 & 6 & B & 0 & D & 4 & 3 & 1 & 6 & D & 0 & F & 5 & A & A & 2 & D & 6 & E & E & B & 4 & 6 & D & 3 & 9 & 8 & B & 7 & 0 & 6 & 3 & 6 & 0 & F & 4 & B & 1 & 5 & 5 & D & 6 & F & 0 & 0 \\
0 & C & C & 0 & 0 & F & B & 8 & 5 & A & E & D & 6 & 0 & C & C & 5 & B & 7 & 0 & B & A & 8 & A & E & 8 & F & 1 & 4 & 3 & C & D & E & 1 & 9 & D & 0 & 0 & C & 0 & D & B & 3 & B & A & E & D & 7 & 3 & 5 & 7 & 6 & B & 5 & B & 8 & 0 & 0 \\
0 & 3 & D & 9 & D & A & A & 0 & 5 & 0 & D & E & 0 & 6 & 3 & E & B & 6 & 0 & 5 & C & F & 3 & D & 8 & 5 & A & B & 5 & B & C & 1 & 8 & A & 1 & D & 6 & 6 & 2 & D & 1 & 7 & 6 & 7 & 6 & E & 6 & 3 & 4 & D & 0 & F & 2 & B & D & 4 & 0 & 0 \\
0 & 0 & 3 & 3 & 0 & 0 & 1 & 7 & 0 & C & D & E & 5 & 6 & F & 3 & 0 & 6 & F & 6 & 0 & D & E & D & 4 & 0 & F & D & 6 & 8 & B & 5 & B & 1 & 5 & B & 6 & 6 & 3 & 3 & D & A & 8 & C & E & B & 0 & F & 3 & 3 & C & 3 & D & D & 1 & A & 8 & 0 \\
0 & 0 & 0 & 7 & B & C & C & D & E & 6 & C & 0 & 6 & B & A & 7 & 6 & B & 8 & 0 & D & 7 & 6 & B & 7 & 5 & A & 0 & C & 3 & 7 & 7 & 3 & A & 3 & 6 & 7 & 8 & C & 5 & D & 4 & 6 & A & C & 5 & 7 & 7 & 4 & 8 & 6 & 3 & D & 8 & 1 & B & 0 & 0 \\
A & 0 & 0 & 0 & 0 & 0 & B & 6 & E & B & 4 & 0 & 6 & C & 0 & 3 & 6 & 6 & E & E & 7 & 6 & E & B & 6 & 8 & 0 & B & B & 5 & 6 & E & E & B & A & E & 0 & 1 & 7 & 6 & D & 7 & 5 & D & 6 & B & A & C & 0 & 6 & 6 & 0 & 2 & B & A & E & C & 8 \\
\end{array}\right].
$}
\]
}
\item A binary self-orthogonal $[230, 11, 110]$ code which has an automorphism group of size $|\mathrm{Aut}|=2$ and a generator matrix
{\scriptsize
\setlength{\arraycolsep}{0.35pt}%
\renewcommand{\arraystretch}{0.9}%
\[
\resizebox{\linewidth}{!}{$\displaystyle G_{\mathrm{SO},28}=
\left[\begin{array}{@{}*{58}{c}@{}}
A & 0 & 0 & 0 & 6 & 9 & 1 & 9 & D & 1 & 5 & E & 6 & D & 9 & A & B & 1 & 9 & E & A & 3 & 3 & 0 & E & D & A & D & 5 & D & 1 & 8 & 3 & 7 & 5 & E & B & 4 & C & 6 & A & 2 & B & 6 & C & 3 & 6 & D & E & 3 & 7 & A & 4 & 6 & A & C & 0 & 0 \\
6 & 0 & 1 & 9 & E & 6 & 7 & 6 & B & 6 & C & 0 & 0 & C & F & 0 & 0 & 7 & B & 0 & D & A & B & 0 & E & E & 9 & 1 & 5 & D & 0 & F & 5 & C & 1 & B & C & D & 5 & E & 0 & C & 6 & D & B & 6 & B & 9 & E & 0 & 0 & C & E & D & 7 & A & 8 & 0 \\
1 & 5 & 4 & 0 & 0 & F & 1 & A & B & C & C & 7 & 8 & 7 & 6 & B & 0 & D & 4 & 3 & 1 & 6 & D & 0 & F & 5 & A & A & 2 & D & 6 & E & E & B & 4 & 6 & D & 3 & 9 & 8 & B & 7 & 0 & 6 & 3 & 6 & 0 & F & 4 & B & 1 & 5 & 5 & D & 6 & F & 0 & 0 \\
0 & C & C & 0 & 0 & F & B & 8 & 5 & A & E & D & 6 & 0 & C & C & 5 & B & 7 & 0 & B & A & 8 & A & E & 8 & F & 1 & 4 & 3 & C & D & E & 1 & 9 & D & 0 & 0 & C & 0 & D & B & 3 & B & A & E & D & 7 & 3 & 5 & 7 & 6 & B & 5 & B & 8 & 0 & 0 \\
0 & 3 & D & 9 & D & A & A & 0 & 5 & 0 & D & E & 0 & 6 & 3 & E & B & 6 & 0 & 5 & C & F & 3 & D & 8 & 5 & A & B & 5 & B & C & 1 & 8 & A & 1 & D & 6 & 6 & 2 & D & 1 & 7 & 6 & 7 & 6 & E & 6 & 3 & 4 & D & 0 & F & 2 & B & D & 4 & 0 & 0 \\
0 & 0 & 3 & 3 & 0 & 0 & 1 & 7 & 0 & C & D & E & 5 & 6 & F & 3 & 0 & 6 & F & 6 & 0 & D & E & D & 4 & 0 & F & D & 6 & 8 & B & 5 & B & 1 & 5 & B & 6 & 6 & 3 & 3 & D & A & 8 & C & E & B & 0 & F & 3 & 3 & C & 3 & D & D & 1 & A & 8 & 0 \\
0 & 0 & 0 & 7 & B & C & C & D & E & 6 & C & 0 & 6 & B & A & 7 & 6 & B & 8 & 0 & D & 7 & 6 & B & 7 & 5 & A & 0 & C & 3 & 7 & 7 & 3 & A & 3 & 6 & 7 & 8 & C & 5 & D & 4 & 6 & A & C & 5 & 7 & 7 & 4 & 8 & 6 & 3 & D & 8 & 1 & B & 0 & 0 \\
A & 0 & 0 & 0 & 0 & 0 & B & 6 & E & B & 4 & 0 & 6 & C & 0 & 3 & 6 & 6 & E & E & 7 & 6 & E & B & 6 & 8 & 0 & B & B & 5 & 6 & E & E & B & A & E & 0 & 1 & 7 & 6 & D & 7 & 5 & D & 6 & B & A & C & 0 & 6 & 6 & 0 & 2 & B & A & E & E & C \\
\end{array}\right].
$}
\]
}
\item A binary self-orthogonal $[230, 11, 110]$ code which has an automorphism group of size $|\mathrm{Aut}|=2$ and a generator matrix
{\scriptsize
\setlength{\arraycolsep}{0.35pt}%
\renewcommand{\arraystretch}{0.9}%
\[
\resizebox{\linewidth}{!}{$\displaystyle G_{\mathrm{SO},29}=
\left[\begin{array}{@{}*{58}{c}@{}}
C & 0 & 0 & 0 & 6 & 7 & 7 & 3 & D & 7 & 5 & 5 & A & E & D & 6 & 6 & C & 0 & B & A & B & 6 & 3 & 6 & 1 & E & F & 5 & C & D & 8 & 6 & 3 & 0 & 6 & 3 & C & 6 & A & 0 & 0 & 7 & 7 & 0 & D & 6 & D & 5 & 7 & 5 & 9 & 8 & 7 & B & 0 & A & 0 \\
2 & A & 8 & 0 & 0 & 0 & D & 8 & 1 & 5 & B & C & A & C & 7 & D & C & 3 & 5 & D & 7 & 3 & 6 & 0 & 6 & 7 & 8 & D & E & 7 & 7 & 5 & 0 & 3 & 5 & C & 3 & D & E & A & E & 8 & A & 1 & E & 1 & 9 & D & A & A & 2 & D & 3 & 1 & 6 & E & 6 & 0 \\
1 & 9 & 8 & 0 & 0 & 1 & 9 & D & D & 9 & E & 9 & 6 & 3 & 3 & 3 & 6 & F & 3 & A & 2 & B & D & 4 & 5 & D & 3 & B & 3 & A & 3 & 3 & 1 & 6 & E & B & 9 & B & 5 & A & 3 & 0 & D & A & D & 0 & 2 & D & 5 & 6 & E & A & E & 0 & C & 0 & 6 & 0 \\
0 & 7 & 8 & 0 & 0 & 0 & 3 & 3 & 7 & 7 & 5 & A & 1 & 9 & E & 0 & D & 6 & B & A & D & E & B & 8 & 3 & B & 4 & 2 & E & C & E & D & 1 & A & D & A & F & 0 & 5 & 7 & B & 0 & C & F & 3 & 1 & 5 & B & 3 & 1 & 5 & 9 & 8 & 1 & 5 & D & C & 0 \\
0 & 0 & 5 & 5 & A & B & B & 0 & A & 3 & 6 & 9 & 1 & A & D & 0 & 6 & E & 8 & 6 & 0 & 5 & D & 8 & 3 & D & 3 & 7 & 3 & 7 & A & D & B & 9 & 8 & 7 & 6 & A & B & 1 & A & D & B & 5 & B & 6 & 0 & 0 & F & 0 & C & 7 & E & 7 & 7 & 0 & A & 0 \\
0 & 0 & 3 & 3 & A & B & 9 & E & B & A & D & 5 & B & 7 & 3 & 6 & 6 & 3 & 6 & A & 1 & E & 0 & 2 & B & C & C & D & E & 6 & C & 6 & D & B & 3 & 1 & A & 6 & D & 6 & 1 & 8 & A & C & 6 & 0 & 1 & E & 9 & 7 & 4 & 0 & 5 & D & 4 & 5 & A & 0 \\
0 & 0 & 0 & F & 6 & 7 & 7 & 6 & A & 1 & B & 3 & 7 & 5 & E & 5 & 0 & E & D & A & F & 0 & 1 & 7 & 6 & 7 & 9 & 8 & 3 & B & 4 & 6 & A & 1 & B & D & 6 & B & 5 & 0 & 2 & E & 7 & 8 & 0 & 7 & 6 & D & A & C & C & 0 & 3 & B & 8 & 6 & 6 & 0 \\
0 & 0 & 0 & 0 & 1 & E & E & D & 6 & C & 3 & 3 & A & 2 & A & B & D & 9 & D & B & 5 & E & C & 1 & 8 & C & D & 9 & B & D & A & 8 & D & 5 & D & 1 & 9 & A & D & B & 8 & 6 & A & D & 8 & D & 8 & 3 & 3 & C & 0 & 7 & 8 & B & B & 0 & 6 & 0 \\
\end{array}\right].
$}
\]
}
\item A binary self-orthogonal $[75, 12, 32]$ code which has an automorphism group of size $|\mathrm{Aut}|=320$ and a generator matrix
{\scriptsize
\setlength{\arraycolsep}{0.35pt}%
\renewcommand{\arraystretch}{0.9}%
\[
G_{\mathrm{SO},30}=
\left[\begin{array}{@{}*{19}{c}@{}}
8 & 1 & 3 & 9 & 0 & 6 & 6 & B & 9 & 2 & 0 & 1 & 3 & C & 5 & E & F & C & 0 \\
4 & 0 & 2 & D & C & A & 6 & F & 7 & 7 & C & F & 1 & F & C & 3 & E & 2 & 0 \\
2 & 1 & 2 & C & 6 & E & F & 6 & 1 & B & 1 & 0 & A & D & B & 3 & 4 & 0 & 0 \\
1 & 0 & 7 & 1 & 2 & B & F & 4 & E & 8 & 3 & 6 & 2 & 2 & F & 9 & 1 & 2 & 0 \\
0 & 9 & 6 & 9 & F & B & D & 8 & D & B & 2 & 2 & F & B & 2 & E & 9 & C & 0 \\
0 & 4 & 2 & A & B & D & D & 9 & 1 & 7 & 0 & 5 & 9 & 5 & 9 & C & A & 2 & 0 \\
0 & 3 & 6 & 7 & 7 & 1 & 8 & F & 7 & B & 7 & 4 & F & 3 & D & 9 & E & 0 & 0 \\
0 & 0 & F & F & C & 0 & 0 & 0 & 7 & F & F & F & F & 8 & 0 & 0 & 0 & C & 0 \\
8 & 0 & 0 & 3 & 4 & D & E & B & 7 & E & E & 1 & A & 1 & 2 & B & 1 & A & 4 \\
4 & 0 & 0 & C & 5 & 3 & 5 & A & 3 & 2 & 7 & 7 & E & 4 & 8 & E & C & 9 & 0 \\
2 & 0 & 0 & B & 9 & A & 8 & 9 & 2 & B & A & 3 & 2 & 8 & C & 6 & A & F & 2 \\
1 & 0 & 0 & 7 & 7 & 8 & 3 & 9 & E & 0 & D & 5 & 7 & 2 & 8 & 3 & 4 & E & E
\end{array}\right].
\]
}
\item A binary self-orthogonal $[76, 12, 32]$ code which has an automorphism group of size $|\mathrm{Aut}|=2$ and a generator matrix
{\scriptsize
\setlength{\arraycolsep}{0.35pt}%
\renewcommand{\arraystretch}{0.9}%
\[
G_{\mathrm{SO},31}=
\left[\begin{array}{@{}*{19}{c}@{}}
8 & 0 & 3 & F & 6 & B & 2 & 1 & E & 7 & 0 & A & 4 & 2 & 6 & 6 & 3 & 9 & 0 \\
4 & 0 & 2 & D & C & A & 6 & F & 7 & 7 & C & F & 1 & F & C & 3 & E & 2 & 0 \\
2 & 0 & 2 & A & 0 & 3 & B & C & 6 & E & 1 & B & D & 3 & 8 & B & 8 & 5 & 0 \\
1 & 0 & 7 & 1 & 2 & B & F & 4 & E & 8 & 3 & 6 & 2 & 2 & F & 9 & 1 & 2 & 0 \\
0 & 8 & 6 & F & 9 & 6 & 9 & 2 & A & E & 2 & 9 & 8 & 5 & 1 & 6 & 5 & 9 & 0 \\
0 & 4 & 2 & A & B & D & D & 9 & 1 & 7 & 0 & 5 & 9 & 5 & 9 & C & A & 2 & 0 \\
0 & 2 & 6 & 1 & 1 & C & C & 5 & 0 & E & 7 & F & 8 & D & E & 1 & 2 & 5 & 0 \\
0 & 1 & 0 & 6 & 6 & D & 4 & A & 7 & 5 & 0 & B & 7 & E & 3 & 8 & C & 5 & 0 \\
0 & 0 & F & F & C & 0 & 0 & 0 & 7 & F & F & F & F & 8 & 0 & 0 & 0 & C & 0 \\
8 & 0 & 0 & 3 & 4 & D & E & B & 7 & E & E & 1 & A & 1 & 2 & B & 1 & A & E \\
4 & 0 & 0 & C & 5 & 3 & 5 & A & 3 & 2 & 7 & 7 & E & 4 & 8 & E & C & 8 & B \\
1 & 0 & 0 & 7 & 7 & 8 & 3 & 9 & E & 0 & D & 5 & 7 & 2 & 8 & 3 & 4 & F & 9
\end{array}\right].
\]
}
\item A binary self-orthogonal $[233, 12, 110]$ code which has an automorphism group of size $|\mathrm{Aut}|=1$ and a generator matrix
{\scriptsize
\setlength{\arraycolsep}{0.35pt}%
\renewcommand{\arraystretch}{0.9}%
\[
\resizebox{\linewidth}{!}{$\displaystyle G_{\mathrm{SO},8}=
\left[\begin{array}{@{}*{59}{c}@{}}
A & 0 & 0 & A & E & E & A & 2 & 8 & A & 1 & D & D & 5 & 8 & B & 9 & D & 1 & 6 & 8 & 6 & D & E & A & D & 8 & 0 & E & 8 & 1 & 6 & D & C & 1 & 8 & A & D & B & B & 6 & B & 7 & 6 & B & D & A & 8 & C & 0 & 3 & 7 & 4 & 5 & D & 7 & 0 & 0 & 0 \\
6 & 0 & 0 & 7 & 9 & D & 6 & 1 & 8 & 6 & 3 & 3 & A & F & 0 & 7 & 7 & 3 & 0 & D & 8 & D & B & 5 & 7 & B & 0 & 1 & 9 & 8 & 0 & D & B & A & 3 & 0 & 7 & B & 6 & 6 & D & E & C & D & E & B & 5 & 8 & A & 0 & 6 & C & C & 3 & A & E & 8 & 0 & 0 \\
1 & 4 & 0 & C & 2 & E & 1 & 7 & 5 & D & 6 & D & C & 0 & 3 & D & 9 & D & B & 9 & D & B & B & 0 & 0 & 3 & 3 & 1 & 6 & 8 & 1 & 9 & B & D & A & B & C & 3 & 6 & A & D & 8 & D & 5 & 8 & 1 & 4 & 5 & 7 & 8 & 5 & D & 6 & E & C & 3 & 6 & 0 & 0 \\
0 & C & 0 & A & 1 & D & 0 & E & B & A & D & B & A & 0 & 6 & B & 7 & 3 & 7 & 7 & 3 & 7 & 6 & 8 & 0 & 2 & E & 0 & D & 8 & 1 & 7 & 6 & B & 5 & E & A & 2 & D & 7 & B & 0 & A & F & 0 & 0 & C & 3 & D & 4 & 3 & A & D & D & A & 2 & D & 0 & 0 \\
0 & 2 & 8 & 7 & 4 & 3 & A & E & B & 6 & E & 8 & 1 & A & D & C & D & D & C & D & D & A & 0 & 0 & B & 8 & 3 & 6 & 3 & 5 & D & A & D & 7 & A & 8 & B & 5 & E & C & 2 & B & C & 0 & 3 & 0 & F & 5 & 0 & E & B & 7 & 6 & 8 & B & 4 & 0 & 0 & 0 \\
0 & 1 & 8 & C & C & 6 & 7 & 9 & E & D & 9 & 8 & 1 & 5 & B & B & B & 3 & B & B & 3 & 6 & 0 & 0 & 7 & 4 & 6 & C & 2 & B & B & 5 & B & D & 5 & 8 & 6 & F & 5 & A & 1 & E & A & 0 & 6 & 1 & A & B & 1 & 9 & E & C & D & 8 & 6 & C & 0 & 0 & 0 \\
0 & 0 & 5 & 7 & 7 & 5 & 1 & 4 & 5 & 0 & E & E & A & C & 5 & C & E & 8 & B & 4 & 3 & 6 & F & 5 & 6 & C & 0 & 7 & 5 & 8 & B & 6 & E & 0 & C & 5 & 6 & D & D & B & 5 & B & B & 5 & E & D & 4 & 6 & 0 & 1 & B & A & 2 & E & B & 8 & 0 & 0 & 0 \\
0 & 0 & 3 & C & E & B & 0 & C & 3 & 1 & 9 & D & 7 & 8 & 3 & B & 9 & 8 & 6 & C & 6 & D & A & B & D & 8 & 0 & C & F & 0 & 6 & D & D & 1 & 8 & 3 & D & B & 3 & 6 & F & 6 & 6 & F & 5 & A & C & 5 & 0 & 3 & 6 & 6 & 1 & D & 7 & 4 & 0 & 0 & 0 \\
\end{array}\right].
$}
\]
}
\item A binary self-orthogonal $[233, 12, 110]$ code which has an automorphism group of size $|\mathrm{Aut}|=1$ and a generator matrix
{\scriptsize
\setlength{\arraycolsep}{0.35pt}%
\renewcommand{\arraystretch}{0.9}%
\[
\resizebox{\linewidth}{!}{$\displaystyle G_{\mathrm{SO},32}=
\left[\begin{array}{@{}*{59}{c}@{}}
A & 0 & 0 & A & E & E & A & 2 & 8 & A & 1 & D & D & 5 & 8 & B & 9 & D & 1 & 6 & 8 & 6 & D & E & A & D & 8 & 0 & E & 8 & 1 & 6 & D & C & 1 & 8 & A & D & B & B & 6 & B & 7 & 6 & B & D & A & 8 & C & 0 & 3 & 7 & 4 & 5 & D & 7 & 0 & 0 & 0 \\
6 & 0 & 0 & 7 & 9 & D & 6 & 1 & 8 & 6 & 3 & 3 & A & F & 0 & 7 & 7 & 3 & 0 & D & 8 & D & B & 5 & 7 & B & 0 & 1 & 9 & 8 & 0 & D & B & A & 3 & 0 & 7 & B & 6 & 6 & D & E & C & D & E & B & 5 & 8 & A & 0 & 6 & C & C & 3 & A & E & 8 & 0 & 0 \\
1 & 4 & 0 & C & 2 & E & 1 & 7 & 5 & D & 6 & D & C & 0 & 3 & D & 9 & D & B & 9 & D & B & B & 0 & 0 & 3 & 3 & 1 & 6 & 8 & 1 & 9 & B & D & A & B & C & 3 & 6 & A & D & 8 & D & 5 & 8 & 1 & 4 & 5 & 7 & 8 & 5 & D & 6 & E & C & 3 & 6 & 0 & 0 \\
0 & C & 0 & A & 1 & D & 0 & E & B & A & D & B & A & 0 & 6 & B & 7 & 3 & 7 & 7 & 3 & 7 & 6 & 8 & 0 & 2 & E & 0 & D & 8 & 1 & 7 & 6 & B & 5 & E & A & 2 & D & 7 & B & 0 & A & F & 0 & 0 & C & 3 & D & 4 & 3 & A & D & D & A & 2 & D & 0 & 0 \\
0 & 2 & 8 & 7 & 4 & 3 & A & E & B & 6 & E & 8 & 1 & A & D & C & D & D & C & D & D & A & 0 & 0 & B & 8 & 3 & 6 & 3 & 5 & D & A & D & 7 & A & 8 & B & 5 & E & C & 2 & B & C & 0 & 3 & 0 & F & 5 & 0 & E & B & 7 & 6 & 8 & B & 4 & 0 & 0 & 0 \\
0 & 1 & 8 & C & C & 6 & 7 & 9 & E & D & 9 & 8 & 1 & 5 & B & B & B & 3 & B & B & 3 & 6 & 0 & 0 & 7 & 4 & 6 & C & 2 & B & B & 5 & B & D & 5 & 8 & 6 & F & 5 & A & 1 & E & A & 0 & 6 & 1 & A & B & 1 & 9 & E & C & D & 8 & 6 & C & 0 & 0 & 0 \\
0 & 0 & 5 & 7 & 7 & 5 & 1 & 4 & 5 & 0 & E & E & A & C & 5 & C & E & 8 & B & 4 & 3 & 6 & F & 5 & 6 & C & 0 & 7 & 5 & 8 & B & 6 & E & 0 & C & 5 & 6 & D & D & B & 5 & B & B & 5 & E & D & 4 & 6 & 0 & 1 & B & A & 2 & E & B & 8 & 0 & 0 & 0 \\
0 & 0 & 3 & C & E & B & 0 & C & 3 & 1 & 9 & D & 7 & 8 & 3 & B & 9 & 8 & 6 & C & 6 & D & A & B & D & 8 & 0 & C & F & 0 & 6 & D & D & 1 & 8 & 3 & D & B & 3 & 6 & F & 6 & 6 & F & 5 & A & C & 5 & 0 & 3 & 6 & 6 & 1 & D & 7 & 4 & 0 & 0 & 0 \\
\end{array}\right].
$}
\]
}
\item A binary self-orthogonal $[81, 15, 32]$ code which has an automorphism group of size $|\mathrm{Aut}|=7$ and a generator matrix
{\scriptsize
\setlength{\arraycolsep}{0.35pt}%
\renewcommand{\arraystretch}{0.9}%
\[
G_{\mathrm{SO},9}=
\left[\begin{array}{@{}*{21}{c}@{}}
8 & 3 & 9 & 8 & 0 & 7 & B & C & B & 3 & 7 & 1 & B & A & 4 & 2 & 3 & 0 & 0 & 0 & 0 \\
4 & 2 & F & 0 & 0 & 2 & 7 & B & 3 & C & 6 & 0 & 9 & 8 & 3 & 7 & E & C & 0 & 0 & 0 \\
2 & 3 & C & 4 & 0 & 3 & F & E & 2 & F & 2 & D & 4 & 7 & 2 & 8 & 4 & 4 & 0 & 0 & 0 \\
1 & 2 & 5 & C & 0 & 7 & 8 & 8 & F & A & 8 & 9 & 0 & 5 & A & 7 & F & 8 & 0 & 0 & 0 \\
0 & A & A & C & 0 & 5 & 0 & 8 & 5 & 5 & B & 3 & 7 & E & 3 & B & C & 4 & 0 & 0 & 0 \\
0 & 7 & 6 & 8 & 0 & 0 & 5 & 1 & 2 & 6 & E & 7 & 6 & F & E & 2 & B & 8 & 0 & 0 & 0 \\
8 & 0 & 0 & 4 & F & B & 7 & 4 & 4 & 0 & 7 & 8 & 4 & 4 & 7 & 3 & 9 & 9 & 5 & 0 & 8 \\
4 & 0 & 0 & 4 & 8 & 9 & F & 5 & 9 & C & C & 2 & 0 & E & A & 5 & B & 8 & 5 & F & 8 \\
2 & 0 & 0 & 4 & 4 & 2 & 9 & 8 & 5 & 7 & 6 & A & 1 & 2 & A & D & E & D & F & 5 & 8 \\
1 & 0 & 0 & 5 & 2 & 8 & 6 & A & 3 & 4 & A & 6 & 8 & 9 & 0 & D & A & E & 8 & B & 8 \\
0 & 4 & 0 & 5 & E & 2 & A & 4 & D & 0 & 1 & 2 & A & 0 & 5 & E & E & C & 6 & 9 & 8 \\
0 & 1 & 0 & 5 & F & 1 & 9 & F & 7 & E & F & 1 & 2 & C & 1 & 0 & E & A & D & 7 & 0 \\
0 & 0 & 8 & 5 & 2 & 1 & 3 & 0 & 3 & F & 4 & 3 & A & D & B & E & 7 & 3 & 5 & C & 0 \\
0 & 0 & 2 & 5 & 3 & 1 & 4 & D & 6 & 1 & 5 & 9 & C & 0 & C & E & 2 & D & 1 & D & 0 \\
0 & 0 & 0 & 3 & F & 8 & 0 & F & E & 0 & 0 & 0 & 7 & F & F & F & F & D & B & F & 8
\end{array}\right].
\]
}
\item A binary self-orthogonal $[82, 15, 32]$ code which has an automorphism group of size $|\mathrm{Aut}|=64$ and a generator matrix
{\scriptsize
\setlength{\arraycolsep}{0.35pt}%
\renewcommand{\arraystretch}{0.9}%
\[
G_{\mathrm{SO},10}=
\left[\begin{array}{@{}*{21}{c}@{}}
8 & 1 & 5 & 0 & 0 & 4 & A & B & 6 & 5 & 6 & F & 3 & 4 & 9 & E & D & 0 & 8 & 0 & 0 \\
4 & 0 & 3 & 8 & 0 & 1 & 6 & C & E & A & 7 & E & 1 & 6 & E & B & 0 & C & 8 & 0 & 0 \\
2 & 1 & 0 & C & 0 & 0 & E & 9 & F & 9 & 3 & 3 & C & 9 & F & 4 & A & 4 & 8 & 0 & 0 \\
1 & 0 & 9 & 4 & 0 & 4 & 9 & F & 2 & C & 9 & 7 & 8 & B & 7 & B & 1 & 8 & 8 & 0 & 0 \\
0 & 8 & 6 & 4 & 0 & 6 & 1 & F & 8 & 3 & A & D & F & 0 & E & 7 & 2 & 4 & 8 & 0 & 0 \\
0 & 5 & A & 0 & 0 & 3 & 4 & 6 & F & 0 & F & 9 & E & 1 & 3 & E & 5 & 8 & 8 & 0 & 0 \\
0 & 2 & C & 8 & 0 & 3 & 1 & 7 & D & 6 & 1 & E & 8 & E & D & C & E & 0 & 8 & 0 & 0 \\
8 & 0 & 0 & 4 & F & B & 7 & 4 & 4 & 0 & 7 & 8 & 4 & 4 & 7 & 3 & 9 & 9 & 8 & 3 & 0 \\
4 & 0 & 0 & 4 & 8 & 9 & F & 5 & 9 & C & C & 2 & 0 & E & A & 5 & B & 8 & 9 & 7 & C \\
2 & 0 & 0 & 4 & 4 & 2 & 9 & 8 & 5 & 7 & 6 & A & 1 & 2 & A & D & E & D & B & B & 8 \\
1 & 0 & 0 & 5 & 2 & 8 & 6 & A & 3 & 4 & A & 6 & 8 & 9 & 0 & D & A & E & A & 3 & 8 \\
0 & 4 & 0 & 5 & E & 2 & A & 4 & D & 0 & 1 & 2 & A & 0 & 5 & E & E & C & E & A & 0 \\
0 & 1 & 0 & 5 & F & 1 & 9 & F & 7 & E & F & 1 & 2 & C & 1 & 0 & E & A & 7 & B & 0 \\
0 & 0 & 8 & 5 & 2 & 1 & 3 & 0 & 3 & F & 4 & 3 & A & D & B & E & 7 & 3 & 3 & C & 0 \\
0 & 0 & 0 & 3 & F & 8 & 0 & F & E & 0 & 0 & 0 & 7 & F & F & F & F & D & 7 & 4 & 0
\end{array}\right].
\]
}
\item A binary self-orthogonal $[82, 15, 32]$ code which has an automorphism group of size $|\mathrm{Aut}|=1$ and a generator matrix
{\scriptsize
\setlength{\arraycolsep}{0.35pt}%
\renewcommand{\arraystretch}{0.9}%
\[
G_{\mathrm{SO},11}=
\left[\begin{array}{@{}*{21}{c}@{}}
8 & 0 & 2 & 0 & C & F & 6 & 5 & 1 & 0 & 1 & 5 & B & D & 8 & 9 & 6 & D & 4 & 8 & 0 \\
4 & 0 & 3 & 8 & C & 4 & 5 & 7 & A & F & 0 & 6 & 5 & A & 9 & A & 1 & 8 & 1 & 8 & 0 \\
2 & 0 & 3 & 9 & 5 & 8 & D & 9 & F & 6 & 5 & 3 & F & B & 9 & 1 & E & 2 & 3 & 8 & 0 \\
1 & 0 & 0 & 7 & 9 & E & 4 & 2 & F & 6 & 9 & 1 & 5 & 8 & C & E & 2 & 8 & 4 & 8 & 0 \\
0 & 8 & 0 & 7 & 4 & 7 & D & 3 & 8 & 6 & F & C & 9 & D & E & 6 & B & F & 2 & 0 & 0 \\
0 & 4 & 3 & 1 & 4 & F & 4 & 7 & 2 & 5 & B & B & 3 & 5 & 7 & E & 0 & 0 & 0 & 0 & 0 \\
0 & 2 & 3 & 8 & 6 & 8 & C & 6 & 8 & 2 & C & 8 & 2 & 7 & 3 & 7 & 6 & D & 4 & 8 & 0 \\
0 & 1 & 1 & 4 & 5 & 8 & E & 3 & 7 & 6 & 7 & D & 2 & A & 4 & D & 7 & 5 & 5 & 0 & 0 \\
0 & 0 & A & 7 & 7 & A & 0 & 3 & A & D & F & 7 & 7 & A & 2 & C & 9 & 7 & 6 & 8 & 0 \\
0 & 0 & 6 & 5 & D & 9 & 5 & D & C & D & 8 & A & F & 7 & 1 & A & B & F & 2 & 0 & 0 \\
8 & 0 & 0 & 0 & D & 4 & D & C & B & B & D & 0 & C & 9 & B & 2 & A & 5 & C & 2 & C \\
4 & 0 & 0 & 4 & E & 2 & 9 & C & A & 0 & 5 & C & 5 & 5 & 5 & 9 & 9 & B & B & C & C \\
2 & 0 & 0 & 4 & D & 7 & A & B & 8 & E & A & 8 & 6 & 7 & 3 & 4 & 5 & 1 & C & B & 0 \\
1 & 0 & 0 & 4 & C & D & 3 & 0 & 1 & 9 & D & 2 & 7 & E & 0 & 3 & 7 & A & 0 & 3 & 8 \\
0 & 4 & 0 & 4 & 7 & 8 & 6 & 9 & D & 1 & 8 & 9 & 5 & 9 & B & E & A & 8 & B & 6 & 4
\end{array}%
\right].
\]
}
\end{itemize}
\endgroup

Therefore, we obtain the following theorem.
\begin{theorem}
There exist binary distance-optimal self-orthogonal codes with the following
parameters:
\[
\begin{gathered}
\relax [53,9,22],\ [61,9,26],\ [69,9,30],\ [84,9,38],\ [109,9,50],\\
[132,9,62],\ [188,9,90],\ [227,9,110],\ [255,9,124],\\
[70,10,30],\ [86,10,38],\ [102,10,46],\ [147,10,68],\
[196,10,94],\\
[188,11,88],\ [189,11,88],\ [191,11,90],\ [198,11,94],\
[230,11,110],\\
[75,12,32],\ [76,12,32],\ [233,12,110],\ [81,15,32],\ [82,15,32].
\end{gathered}
\]
\end{theorem}

Among all self-orthogonal codes obtained in our experiments, the parameters of those whose minimum distance attains the bound and those whose minimum distance is two below the bound are presented in Tables~\ref{tab:so-embedding-gap0-ineq} and~\ref{tab:so-embedding-gap2-ineq}, respectively. For each parameters, $d_{\mathrm{Found}}$ denotes the largest minimum distance attained by the codes obtained in our experiments, and $N_{\mathrm{ineq}}$ denotes the number of inequivalent codes attaining this minimum distance. In the $d_{\mathrm{Found}}$ column of Table~\ref{tab:so-embedding-gap0-ineq}, boldface values indicate parameters that were not previously known to be attainable.
\begingroup
\footnotesize
\renewcommand{\arraystretch}{1.3}
\setlength{\tabcolsep}{4pt}
\begin{longtable}{c|c|c||c|c|c||c|c|c||c|c|c}
\caption{Self-orthogonal embeddings of codes from BKLC library which attain the bounds}
\label{tab:so-embedding-gap0-ineq}\\
\hline
$[n,k]$ & $d_{\mathrm{Found}}$ & $N_{\mathrm{ineq}}$ & $[n,k]$ & $d_{\mathrm{Found}}$ & $N_{\mathrm{ineq}}$ & $[n,k]$ & $d_{\mathrm{Found}}$ & $N_{\mathrm{ineq}}$ & $[n,k]$ & $d_{\mathrm{Found}}$ & $N_{\mathrm{ineq}}$ \\
\hline
\endfirsthead
\multicolumn{12}{c}{}\\
\hline
$[n,k]$ & $d_{\mathrm{Found}}$ & $N_{\mathrm{ineq}}$ & $[n,k]$ & $d_{\mathrm{Found}}$ & $N_{\mathrm{ineq}}$ & $[n,k]$ & $d_{\mathrm{Found}}$ & $N_{\mathrm{ineq}}$ & $[n,k]$ & $d_{\mathrm{Found}}$ & $N_{\mathrm{ineq}}$ \\
\hline
\endhead
\multicolumn{12}{r}{}\\
\endfoot
\endlastfoot
$[43,9]$ & 16 & 1 & $[254,9]$ & 124 & 2 & $[58,11]$ & 24 & 1 & $[158,12]$ & 72 & 1 \\
\hline
$[46,9]$ & 20 & 1 & $[255,9]$ & \textbf{124} & 1 & $[59,11]$ & 24 & 1 & $[174,12]$ & 80 & 1 \\
\hline
$[49,9]$ & 20 & 1 & $[256,9]$ & 128 & 1 & $[66,11]$ & 28 & 1 & $[189,12]$ & 88 & 1 \\
\hline
$[50,9]$ & 20 & 1 & $[42,10]$ & 16 & 1 & $[72,11]$ & 32 & 1 & $[190,12]$ & 88 & 1 \\
\hline
$[53,9]$ & \textbf{22} & 1 & $[43,10]$ & 16 & 1 & $[76,11]$ & 32 & 1 & $[233,12]$ & \textbf{110} & 2 \\
\hline
$[54,9]$ & 24 & 2 & $[48,10]$ & 20 & 1 & $[77,11]$ & 32 & 1 & $[234,12]$ & 112 & 1 \\
\hline
$[55,9]$ & 24 & 1 & $[49,10]$ & 20 & 2 & $[83,11]$ & 36 & 1 & $[252,12]$ & 120 & 1 \\
\hline
$[61,9]$ & \textbf{26} & 1 & $[50,10]$ & 20 & 1 & $[89,11]$ & 40 & 1 & $[253,12]$ & 120 & 1 \\
\hline
$[62,9]$ & 28 & 1 & $[56,10]$ & 24 & 2 & $[99,11]$ & 44 & 1 & $[26,13]$ & 6 & 1 \\
\hline
$[63,9]$ & 28 & 1 & $[57,10]$ & 24 & 1 & $[106,11]$ & 48 & 1 & $[30,13]$ & 8 & 1 \\
\hline
$[69,9]$ & \textbf{30} & 1 & $[58,10]$ & 24 & 1 & $[122,11]$ & 56 & 2 & $[31,13]$ & 8 & 1 \\
\hline
$[70,9]$ & 32 & 2 & $[64,10]$ & 28 & 1 & $[132,11]$ & 60 & 1 & $[36,13]$ & 12 & 1 \\
\hline
$[71,9]$ & 32 & 1 & $[65,10]$ & 28 & 1 & $[138,11]$ & 64 & 1 & $[39,13]$ & 12 & 1 \\
\hline
$[79,9]$ & 36 & 1 & $[66,10]$ & 28 & 1 & $[139,11]$ & 64 & 1 & $[44,13]$ & 16 & 1 \\
\hline
$[84,9]$ & \textbf{38} & 1 & $[70,10]$ & \textbf{30} & 1 & $[148,11]$ & 68 & 1 & $[53,13]$ & 20 & 3 \\
\hline
$[86,9]$ & 40 & 2 & $[71,10]$ & 32 & 1 & $[155,11]$ & 72 & 1 & $[60,13]$ & 24 & 1 \\
\hline
$[87,9]$ & 40 & 1 & $[81,10]$ & 36 & 1 & $[157,11]$ & 72 & 1 & $[61,13]$ & 24 & 1 \\
\hline
$[95,9]$ & 44 & 1 & $[86,10]$ & \textbf{38} & 1 & $[171,11]$ & 80 & 1 & $[70,13]$ & 28 & 1 \\
\hline
$[96,9]$ & 44 & 1 & $[87,10]$ & 40 & 1 & $[172,11]$ & 80 & 1 & $[78,13]$ & 32 & 3 \\
\hline
$[102,9]$ & 48 & 2 & $[96,10]$ & 44 & 1 & $[188,11]$ & \textbf{88} & 3 & $[79,13]$ & 32 & 1 \\
\hline
$[103,9]$ & 48 & 1 & $[102,10]$ & \textbf{46} & 1 & $[189,11]$ & \textbf{88} & 3 & $[80,13]$ & 32 & 1 \\
\hline
$[109,9]$ & \textbf{50} & 1 & $[103,10]$ & 48 & 1 & $[191,11]$ & \textbf{90} & 1 & $[95,13]$ & 40 & 1 \\
\hline
$[110,9]$ & 52 & 1 & $[112,10]$ & 52 & 2 & $[192,11]$ & 92 & 1 & $[111,13]$ & 48 & 1 \\
\hline
$[111,9]$ & 52 & 1 & $[119,10]$ & 56 & 2 & $[196,11]$ & 92 & 1 & $[126,13]$ & 56 & 1 \\
\hline
$[118,9]$ & 56 & 2 & $[128,10]$ & 60 & 1 & $[198,11]$ & \textbf{94} & 1 & $[127,13]$ & 56 & 1 \\
\hline
$[119,9]$ & 56 & 1 & $[130,10]$ & 60 & 1 & $[199,11]$ & 96 & 1 & $[143,13]$ & 64 & 1 \\
\hline
$[126,9]$ & 60 & 1 & $[136,10]$ & 64 & 2 & $[210,11]$ & 100 & 1 & $[176,13]$ & 80 & 1 \\
\hline
$[132,9]$ & \textbf{62} & 1 & $[147,10]$ & \textbf{68} & 1 & $[217,11]$ & 104 & 1 & $[191,13]$ & 88 & 1 \\
\hline
$[133,9]$ & 64 & 1 & $[152,10]$ & 72 & 1 & $[224,11]$ & 108 & 1 & $[192,13]$ & 88 & 1 \\
\hline
$[143,9]$ & 68 & 1 & $[153,10]$ & 72 & 1 & $[230,11]$ & \textbf{110} & 3 & $[208,13]$ & 96 & 1 \\
\hline
$[146,9]$ & 68 & 1 & $[161,10]$ & 76 & 1 & $[231,11]$ & 112 & 1 & $[240,13]$ & 112 & 1 \\
\hline
$[150,9]$ & 72 & 1 & $[168,10]$ & 80 & 1 & $[251,11]$ & 120 & 2 & $[256,13]$ & 120 & 1 \\
\hline
$[159,9]$ & 76 & 2 & $[177,10]$ & 84 & 1 & $[252,11]$ & 120 & 1 & $[30,14]$ & 8 & 1 \\
\hline
$[160,9]$ & 76 & 1 & $[178,10]$ & 84 & 1 & $[24,12]$ & 8 & 1 & $[31,14]$ & 8 & 2 \\
\hline
$[163,9]$ & 78 & 1 & $[184,10]$ & 88 & 1 & $[26,12]$ & 8 & 1 & $[32,14]$ & 8 & 1 \\
\hline
$[165,9]$ & 80 & 1 & $[191,10]$ & 92 & 1 & $[27,12]$ & 8 & 2 & $[37,14]$ & 12 & 1 \\
\hline
$[175,9]$ & 84 & 2 & $[196,10]$ & \textbf{94} & 1 & $[28,12]$ & 8 & 1 & $[45,14]$ & 16 & 1 \\
\hline
$[181,9]$ & 88 & 1 & $[198,10]$ & 96 & 2 & $[29,12]$ & 8 & 1 & $[54,14]$ & 20 & 2 \\
\hline
$[188,9]$ & \textbf{90} & 1 & $[199,10]$ & 96 & 1 & $[30,12]$ & 8 & 1 & $[62,14]$ & 24 & 1 \\
\hline
$[189,9]$ & 92 & 1 & $[209,10]$ & 100 & 1 & $[34,12]$ & 12 & 1 & $[63,14]$ & 24 & 1 \\
\hline
$[190,9]$ & 92 & 1 & $[210,10]$ & 100 & 1 & $[43,12]$ & 16 & 1 & $[72,14]$ & 28 & 1 \\
\hline
$[195,9]$ & 94 & 1 & $[216,10]$ & 104 & 1 & $[44,12]$ & 16 & 1 & $[73,14]$ & 28 & 1 \\
\hline
$[196,9]$ & 96 & 1 & $[217,10]$ & 104 & 1 & $[46,12]$ & 16 & 1 & $[81,14]$ & 32 & 2 \\
\hline
$[207,9]$ & 100 & 1 & $[223,10]$ & 108 & 1 & $[47,12]$ & 16 & 1 & $[112,14]$ & 48 & 1 \\
\hline
$[208,9]$ & 100 & 2 & $[224,10]$ & 108 & 1 & $[52,12]$ & 20 & 2 & $[127,14]$ & 56 & 1 \\
\hline
$[214,9]$ & 104 & 1 & $[230,10]$ & 112 & 1 & $[53,12]$ & 20 & 1 & $[192,14]$ & 88 & 1 \\
\hline
$[215,9]$ & 104 & 2 & $[231,10]$ & 112 & 1 & $[59,12]$ & 24 & 1 & $[193,14]$ & 88 & 1 \\
\hline
$[216,9]$ & 104 & 1 & $[249,10]$ & 120 & 3 & $[60,12]$ & 24 & 1 & $[31,15]$ & 8 & 1 \\
\hline
$[221,9]$ & 108 & 2 & $[256,10]$ & 124 & 1 & $[68,12]$ & 28 & 1 & $[32,15]$ & 8 & 2 \\
\hline
$[222,9]$ & 108 & 1 & $[22,11]$ & 6 & 1 & $[75,12]$ & \textbf{32} & 1 & $[33,15]$ & 8 & 1 \\
\hline
$[223,9]$ & 108 & 1 & $[23,11]$ & 8 & 1 & $[76,12]$ & \textbf{32} & 1 & $[39,15]$ & 12 & 1 \\
\hline
$[227,9]$ & \textbf{110} & 1 & $[24,11]$ & 8 & 1 & $[77,12]$ & 32 & 1 & $[47,15]$ & 16 & 1 \\
\hline
$[228,9]$ & 112 & 2 & $[25,11]$ & 8 & 1 & $[78,12]$ & 32 & 1 & $[55,15]$ & 20 & 1 \\
\hline
$[229,9]$ & 112 & 1 & $[27,11]$ & 8 & 1 & $[85,12]$ & 36 & 1 & $[63,15]$ & 24 & 1 \\
\hline
$[238,9]$ & 116 & 1 & $[28,11]$ & 8 & 1 & $[91,12]$ & 40 & 1 & $[64,15]$ & 24 & 1 \\
\hline
$[245,9]$ & 120 & 1 & $[32,11]$ & 12 & 1 & $[109,12]$ & 48 & 1 & $[81,15]$ & \textbf{32} & 1 \\
\hline
$[246,9]$ & 120 & 1 & $[41,11]$ & 16 & 1 & $[110,12]$ & 48 & 1 & $[82,15]$ & \textbf{32} & 4 \\
\hline
$[252,9]$ & 124 & 1 & $[50,11]$ & 20 & 1 & $[125,12]$ & 56 & 2 & $[128,15]$ & 56 & 2 \\
\hline
$[253,9]$ & 124 & 1 & $[51,11]$ & 20 & 1 & $[140,12]$ & 64 & 1 & $[194,15]$ & 88 & 1 \\
\hline
\end{longtable}
\endgroup

\begingroup
\footnotesize
\renewcommand{\arraystretch}{1.3}
\setlength{\tabcolsep}{4pt}
\begin{longtable}{c|c|c||c|c|c||c|c|c||c|c|c}
\caption{Self-orthogonal embeddings of codes from BKLC library which lie two below the bounds}
\label{tab:so-embedding-gap2-ineq}\\
\hline
$[n,k]$ & $d_{\mathrm{Found}}$ & $N_{\mathrm{ineq}}$ & $[n,k]$ & $d_{\mathrm{Found}}$ & $N_{\mathrm{ineq}}$ & $[n,k]$ & $d_{\mathrm{Found}}$ & $N_{\mathrm{ineq}}$ & $[n,k]$ & $d_{\mathrm{Found}}$ & $N_{\mathrm{ineq}}$ \\
\hline
\endfirsthead
\multicolumn{12}{c}{}\\
\hline
$[n,k]$ & $d_{\mathrm{Found}}$ & $N_{\mathrm{ineq}}$ & $[n,k]$ & $d_{\mathrm{Found}}$ & $N_{\mathrm{ineq}}$ & $[n,k]$ & $d_{\mathrm{Found}}$ & $N_{\mathrm{ineq}}$ & $[n,k]$ & $d_{\mathrm{Found}}$ & $N_{\mathrm{ineq}}$ \\
\hline
\endhead
\multicolumn{12}{r}{}\\
\endfoot
\endlastfoot
$[44,9]$ & 16 & 1 & $[98,11]$ & 42 & 1 & $[251,12]$ & 118 & 2 & $[80,14]$ & 30 & 1 \\
\hline
$[52,9]$ & 20 & 1 & $[104,11]$ & 44 & 1 & $[252,12]$ & 118 & 1 & $[82,14]$ & 32 & 3 \\
\hline
$[83,9]$ & 36 & 1 & $[105,11]$ & 46 & 1 & $[27,13]$ & 6 & 1 & $[83,14]$ & 32 & 1 \\
\hline
$[84,9]$ & 36 & 1 & $[119,11]$ & 52 & 1 & $[28,13]$ & 6 & 2 & $[90,14]$ & 36 & 1 \\
\hline
$[85,9]$ & 38 & 1 & $[121,11]$ & 54 & 3 & $[29,13]$ & 6 & 1 & $[91,14]$ & 36 & 1 \\
\hline
$[94,9]$ & 42 & 1 & $[131,11]$ & 58 & 1 & $[30,13]$ & 6 & 1 & $[97,14]$ & 38 & 1 \\
\hline
$[101,9]$ & 46 & 1 & $[137,11]$ & 62 & 1 & $[33,13]$ & 8 & 1 & $[98,14]$ & 40 & 1 \\
\hline
$[117,9]$ & 54 & 1 & $[147,11]$ & 66 & 1 & $[34,13]$ & 8 & 2 & $[99,14]$ & 40 & 1 \\
\hline
$[125,9]$ & 58 & 1 & $[149,11]$ & 68 & 1 & $[35,13]$ & 8 & 1 & $[106,14]$ & 44 & 1 \\
\hline
$[148,9]$ & 68 & 1 & $[154,11]$ & 70 & 1 & $[38,13]$ & 10 & 1 & $[107,14]$ & 44 & 1 \\
\hline
$[149,9]$ & 70 & 2 & $[158,11]$ & 72 & 1 & $[40,13]$ & 12 & 1 & $[146,14]$ & 64 & 1 \\
\hline
$[150,9]$ & 70 & 1 & $[164,11]$ & 74 & 1 & $[41,13]$ & 12 & 1 & $[147,14]$ & 64 & 1 \\
\hline
$[174,9]$ & 82 & 1 & $[165,11]$ & 76 & 1 & $[46,13]$ & 14 & 1 & $[188,14]$ & 84 & 2 \\
\hline
$[180,9]$ & 86 & 1 & $[170,11]$ & 78 & 1 & $[50,13]$ & 16 & 1 & $[189,14]$ & 84 & 2 \\
\hline
$[189,9]$ & 90 & 1 & $[187,11]$ & 86 & 1 & $[52,13]$ & 18 & 1 & $[191,14]$ & 86 & 1 \\
\hline
$[210,9]$ & 100 & 1 & $[190,11]$ & 88 & 3 & $[60,13]$ & 22 & 1 & $[192,14]$ & 86 & 1 \\
\hline
$[237,9]$ & 114 & 1 & $[209,11]$ & 98 & 2 & $[73,13]$ & 28 & 1 & $[204,14]$ & 92 & 1 \\
\hline
$[248,9]$ & 120 & 1 & $[216,11]$ & 102 & 2 & $[79,13]$ & 30 & 1 & $[209,14]$ & 94 & 1 \\
\hline
$[250,9]$ & 120 & 1 & $[223,11]$ & 106 & 2 & $[80,13]$ & 30 & 1 & $[210,14]$ & 96 & 2 \\
\hline
$[251,9]$ & 120 & 1 & $[253,11]$ & 120 & 1 & $[87,13]$ & 34 & 1 & $[211,14]$ & 96 & 1 \\
\hline
$[43,10]$ & 14 & 1 & $[25,12]$ & 6 & 1 & $[88,13]$ & 36 & 1 & $[228,14]$ & 104 & 2 \\
\hline
$[84,10]$ & 36 & 1 & $[31,12]$ & 8 & 1 & $[89,13]$ & 36 & 1 & $[229,14]$ & 104 & 1 \\
\hline
$[85,10]$ & 36 & 1 & $[32,12]$ & 8 & 1 & $[94,13]$ & 38 & 1 & $[30,15]$ & 6 & 1 \\
\hline
$[95,10]$ & 42 & 1 & $[33,12]$ & 8 & 1 & $[96,13]$ & 40 & 1 & $[34,15]$ & 8 & 2 \\
\hline
$[111,10]$ & 50 & 2 & $[39,12]$ & 12 & 1 & $[103,13]$ & 42 & 1 & $[35,15]$ & 8 & 1 \\
\hline
$[118,10]$ & 54 & 2 & $[40,12]$ & 12 & 1 & $[104,13]$ & 44 & 1 & $[36,15]$ & 8 & 1 \\
\hline
$[135,10]$ & 62 & 1 & $[41,12]$ & 12 & 1 & $[105,13]$ & 44 & 1 & $[37,15]$ & 8 & 1 \\
\hline
$[148,10]$ & 68 & 1 & $[44,12]$ & 14 & 1 & $[110,13]$ & 46 & 1 & $[39,15]$ & 10 & 1 \\
\hline
$[149,10]$ & 68 & 1 & $[51,12]$ & 18 & 1 & $[114,13]$ & 48 & 1 & $[41,15]$ & 10 & 1 \\
\hline
$[151,10]$ & 70 & 1 & $[58,12]$ & 22 & 1 & $[142,13]$ & 62 & 1 & $[48,15]$ & 14 & 1 \\
\hline
$[157,10]$ & 72 & 1 & $[78,12]$ & 30 & 1 & $[146,13]$ & 64 & 1 & $[50,15]$ & 14 & 1 \\
\hline
$[160,10]$ & 74 & 1 & $[79,12]$ & 32 & 1 & $[155,13]$ & 68 & 1 & $[51,15]$ & 16 & 1 \\
\hline
$[167,10]$ & 78 & 1 & $[100,12]$ & 42 & 1 & $[159,13]$ & 70 & 1 & $[52,15]$ & 16 & 1 \\
\hline
$[176,10]$ & 82 & 1 & $[102,12]$ & 44 & 1 & $[160,13]$ & 72 & 2 & $[53,15]$ & 16 & 1 \\
\hline
$[183,10]$ & 86 & 1 & $[103,12]$ & 44 & 1 & $[161,13]$ & 72 & 1 & $[58,15]$ & 18 & 1 \\
\hline
$[188,10]$ & 88 & 1 & $[108,12]$ & 46 & 1 & $[176,13]$ & 78 & 2 & $[63,15]$ & 22 & 1 \\
\hline
$[189,10]$ & 88 & 1 & $[124,12]$ & 54 & 2 & $[177,13]$ & 80 & 1 & $[74,15]$ & 26 & 1 \\
\hline
$[190,10]$ & 90 & 2 & $[134,12]$ & 58 & 1 & $[190,13]$ & 86 & 2 & $[77,15]$ & 28 & 1 \\
\hline
$[191,10]$ & 90 & 1 & $[135,12]$ & 60 & 1 & $[191,13]$ & 86 & 1 & $[83,15]$ & 32 & 2 \\
\hline
$[208,10]$ & 98 & 1 & $[150,12]$ & 66 & 1 & $[203,13]$ & 92 & 1 & $[84,15]$ & 32 & 1 \\
\hline
$[215,10]$ & 102 & 1 & $[151,12]$ & 68 & 1 & $[207,13]$ & 94 & 1 & $[85,15]$ & 32 & 2 \\
\hline
$[222,10]$ & 106 & 1 & $[152,12]$ & 68 & 1 & $[209,13]$ & 96 & 1 & $[90,15]$ & 34 & 1 \\
\hline
$[229,10]$ & 110 & 1 & $[156,12]$ & 70 & 1 & $[210,13]$ & 96 & 1 & $[91,15]$ & 36 & 1 \\
\hline
$[244,10]$ & 116 & 1 & $[158,12]$ & 70 & 1 & $[220,13]$ & 100 & 1 & $[98,15]$ & 38 & 1 \\
\hline
$[248,10]$ & 118 & 2 & $[159,12]$ & 72 & 1 & $[224,13]$ & 102 & 2 & $[99,15]$ & 40 & 1 \\
\hline
$[251,10]$ & 120 & 1 & $[168,12]$ & 76 & 1 & $[226,13]$ & 104 & 1 & $[106,15]$ & 42 & 1 \\
\hline
$[255,10]$ & 122 & 1 & $[172,12]$ & 78 & 1 & $[227,13]$ & 104 & 2 & $[107,15]$ & 44 & 1 \\
\hline
$[31,11]$ & 8 & 1 & $[175,12]$ & 80 & 1 & $[28,14]$ & 6 & 1 & $[108,15]$ & 44 & 1 \\
\hline
$[37,11]$ & 12 & 1 & $[188,12]$ & 86 & 1 & $[33,14]$ & 8 & 1 & $[116,15]$ & 48 & 1 \\
\hline
$[39,11]$ & 12 & 1 & $[189,12]$ & 86 & 1 & $[34,14]$ & 8 & 2 & $[117,15]$ & 48 & 1 \\
\hline
$[42,11]$ & 14 & 1 & $[201,12]$ & 92 & 1 & $[35,14]$ & 8 & 2 & $[147,15]$ & 64 & 1 \\
\hline
$[43,11]$ & 14 & 1 & $[206,12]$ & 94 & 2 & $[38,14]$ & 10 & 1 & $[149,15]$ & 64 & 1 \\
\hline
$[46,11]$ & 16 & 1 & $[207,12]$ & 96 & 3 & $[40,14]$ & 10 & 1 & $[190,15]$ & 84 & 2 \\
\hline
$[49,11]$ & 18 & 1 & $[208,12]$ & 96 & 1 & $[41,14]$ & 12 & 1 & $[193,15]$ & 86 & 2 \\
\hline
$[57,11]$ & 22 & 1 & $[229,12]$ & 106 & 3 & $[47,14]$ & 14 & 1 & $[194,15]$ & 86 & 1 \\
\hline
$[65,11]$ & 26 & 1 & $[230,12]$ & 108 & 2 & $[48,14]$ & 14 & 1 & $[195,15]$ & 88 & 1 \\
\hline
$[78,11]$ & 32 & 1 & $[231,12]$ & 108 & 2 & $[62,14]$ & 22 & 1 & $[213,15]$ & 96 & 1 \\
\hline
$[82,11]$ & 34 & 1 & $[232,12]$ & 108 & 2 & $[74,14]$ & 28 & 1 & & & \\
\hline
\end{longtable}
\endgroup

\begin{remark}
    H. Shimizu~\cite{H-2026} has reported us that there is no self-orthogonal $[52, 9, 22]$ code. It was shown that if a self-orthogonal $[52,9,22]$ code existed, then it would contain a doubly even $[52,8,24]$ code as a subcode. The author classified doubly even $[52,8,24]$ codes and showed that there are exactly three equivalence classes. The nonexistence result was then established by proving that none of these three classes can be extended to a self-orthogonal $[52,9,22]$ code. Therefore, the $[52, 9, 20]$ code we found in Table~\ref{tab:so-embedding-gap2-ineq} is distance-optimal.
\end{remark}

The generator matrices of the resulting self-orthogonal embeddings are available at~\cite{L-online}.

\section{Conclusion}
In this work, we developed a genetic algorithm for finding shortest self-orthogonal embeddings with large minimum distance. Since every shortest self-orthogonal embedding can be obtained from a fixed one by the action of the orthogonal group, and since the orthogonal group is generated by transvections associated with weight-four vectors, we represent each chromosome as a sequence of weight-four vectors. Computational experiments for $9\le k\le 15$ and $22\le n\le 256$ produced 280 optimal self-orthogonal codes. Among these, 34 have parameters not previously reported, while 100 are inequivalent to known codes with the same parameters. These results demonstrate that the proposed genetic approach provides an effective way to search the space of shortest self-orthogonal embeddings and to construct new optimal self-orthogonal codes.

Although this paper focuses on binary self-orthogonal codes, embedding problems have already been extended to codes over finite fields and finite rings, as well as to embeddings into codes with arbitrary hull dimensions. Therefore, developing algorithmic methods for finding embeddings with the largest possible minimum distance in these more general settings would be a natural direction for future research.

\section*{Acknowledgements}
J.-L. Kim was supported in part by the BK21 FOUR (Fostering Outstanding Universities for Research) funded by the Ministry of Education (MOE, Korea), National Research Foundation of Korea (NRF) under Grant No. 4120240415042, Basic Science Research Program through the National Research Foundation of Korea (NRF) funded by the Ministry of Science and ICT under Grant No. RS-2025-24534992 and Global - Learning \& Academic research institution for Master’s·PhD students, and Postdocs(LAMP) Program of the National Research Foundation of Korea(NRF) grant funded by the Ministry of Education(No. RS-2024-00441954).

The authors would like to thank H. Shimizu for pointing out that a binary optimal self-orthogonal $[52, 9]$ code has $d=20$.

\end{document}